\documentclass[conference]{IEEEtran}
\usepackage{setspace} 
\usepackage{tabularx}
\usepackage{graphicx}
\usepackage{balance}  

\def \figPath {./pic_corr/}

\usepackage{cite}
\usepackage{graphicx}
\usepackage{amsbsy}
\usepackage{amssymb}
\usepackage{amsfonts}
\usepackage[procnumbered,linesnumbered,ruled,vlined]{algorithm2e}
\usepackage{bm}
\usepackage{array}
\usepackage{tabularx}
\usepackage[font=bf]{caption}
\usepackage[hyphens]{url}
\usepackage{amsmath}
\usepackage{environ}

\usepackage{outlines}
\usepackage{soul}
\setul{0.2ex}{0.1ex}
\usepackage[T1]{fontenc}
\usepackage{cases}
\usepackage{diagbox}

\usepackage{amsthm}
\newtheorem{lemma}{Lemma}
\newtheorem{theorem}{Theorem}
\newtheorem{definition}{Definition}

\newtheorem{example}{Example}
\newtheorem{corollary}{Corollary}

\newtheorem{remark}{Remark}

\makeatletter
\newcommand*{\indep}{%
  \mathbin{%
    \mathpalette{\@indep}{}%
  }%
}
\newcommand*{\nindep}{%
  \mathbin{
    \mathpalette{\@indep}{\not}
  }%
}
\newcommand*{\@indep}[2]{%
  \sbox0{$#1\perp\m@th$}
  \sbox2{$#1=$}
  \sbox4{$#1\vcenter{}$}
  \rlap{\copy0}
  \dimen@=\dimexpr\ht2-\ht4-.2pt\relax
  \kern\dimen@
  {#2}%
  \kern\dimen@
  \copy0 
} 

\newcommand{\eqdef}{\overset{\mathrm{def}}{=\joinrel=}}

\newenvironment{alignSmall}{\nobreak\small\noindent\align}{\endalign}
\newenvironment{alignFootnotesize}{\nobreak\footnotesize\noindent\align}{\endalign}
\newenvironment{alignScriptsize}{\nobreak\scriptsize\noindent\align}{\endalign}

\NewEnviron{myAlignS}{%
\begin{singlespace}
\vspace{-8pt}
\begin{alignSmall}
  \BODY
\end{alignSmall}
\vspace{-5pt}
\end{singlespace}
}

\NewEnviron{myAlignSS}{%
\begin{singlespace}
\vspace{-8pt}
\begin{alignFootnotesize}
  \BODY
\end{alignFootnotesize}
\vspace{-5pt}
\end{singlespace}
}

\NewEnviron{myAlignSSS}{%
\begin{singlespace}
\vspace{-8pt}
\begin{alignScriptsize}
  \BODY
\end{alignScriptsize}
\vspace{-5pt}
\end{singlespace}
}

\begin{document}

%

\title{Quantifying Differential Privacy under Temporal Correlations}

\author{\IEEEauthorblockN{Yang Cao\IEEEauthorrefmark{1}\IEEEauthorrefmark{2},
Masatoshi Yoshikawa\IEEEauthorrefmark{1},
Yonghui Xiao\IEEEauthorrefmark{2}, 
Li Xiong\IEEEauthorrefmark{2}}
\IEEEauthorblockA{\IEEEauthorrefmark{1}Department of Social Informatics, Kyoto University, Kyoto, Japan
\\ Email: \{soyo@db.soc.,  yoshikawa@\}i.kyoto--u.ac.jp }
\IEEEauthorblockA{\IEEEauthorrefmark{2}Department of Math and Computer Science, Emory University, Atlanta, USA\\
Email: \{ycao31, yonghui.xiao, lxiong\}@emory.edu}
}

\maketitle
\begin{abstract}
Differential Privacy (DP) has received increasing attention as a rigorous privacy framework.
Many existing studies employ traditional DP mechanisms (e.g., the Laplace mechanism) as primitives, which assume that the data are independent, or that adversaries do not have knowledge of the data correlations.
However, continuous generated data in the real world tend to be temporally correlated, and such correlations can be acquired by adversaries.
In this paper, we investigate the potential privacy loss of a traditional DP mechanism under temporal correlations in the context of continuous data release.
First, we model the temporal correlations using Markov model and analyze the privacy leakage of a DP mechanism when adversaries have knowledge of such temporal correlations.
Our analysis reveals that the privacy loss of a DP mechanism may \textit{accumulate and increase over time}.
We call it \textit{temporal privacy leakage}.
Second, to measure such privacy loss, we design an efficient algorithm for calculating it in polynomial time. 
Although the temporal privacy leakage may increase over time, we also show that its supremum may exist in some cases.
Third, to bound the privacy loss, we propose mechanisms that convert any existing DP mechanism into one against temporal privacy leakage.
Experiments with synthetic data confirm that our approach is efficient and effective. 
\end{abstract}

\section{Introduction}
\label{sec:intro}
With the development of wearable and mobile devices, vast amount of temporal data generated by individuals are being collected, such as trajectories and web page click streams. 
The {continuous publication} of statistics from these temporal data has the potential for significant social benefits such as disease surveillance\cite{bradley_biosense:_2005}, real-time traffic monitoring\cite{federal_highway_administration_fhwa_traffic_2013} and web mining\cite{kosala_web_2000}.
However, privacy concerns hinder the wider use of these data.
To this end, \textit{differential privacy under continual observation}
\cite{acs_case_2014}
\cite{bolot_private_2013}
\cite{chan_private_2011}
\cite{dwork_differential_2010-2}
\cite{dwork_differential_2010}
\cite{fan_fast:_2013}
\cite{kellaris_differentially_2014}
\cite{xiao_dpcube:_2014} 
has received increasing attention because it provides a rigorous privacy guarantee.
Intuitively, differential privacy (DP)\cite{dwork_differential_2006} ensures that the modification of any single user's data in the database has a ``slight'' (bounded in $ \epsilon $) impact on the change in outputs.
The parameter $ \epsilon $ is defined to be a positive real number to control the privacy level.
Larger values of $ \epsilon $ result in larger privacy leakage.
\begin{figure}[t]
\centering
\includegraphics[scale=0.42]{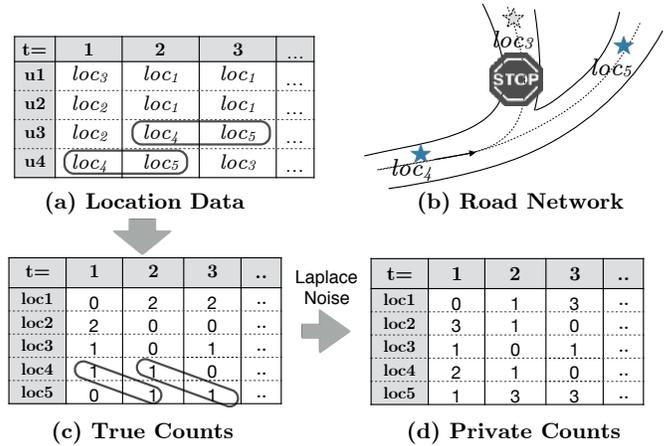} 
\vspace{-1pt}
\caption{Differentially Private Continuous Aggregate Release under Temporal Correlations.}
\label{fig:corr_example}
\vspace{-2pt}
\end{figure}

However, 
most existing works on differentially private continuous aggregate release has an implicit assumption of data independence, i.e., there is no correlation between the data. 
Recent studies\cite{kifer_no_2011}\cite{kifer_rigorous_2012}\cite{kifer_pufferfish:_2014} point out that traditional DP may not guarantee the expected privacy on correlated data.
The following example shows that the temporal correlations may degrade the expected privacy guarantee of  DP.

\begin{example}
\label{example}
Consider the scenario of continuous aggregate release illustrated in Figure \ref{fig:corr_example}.
A trusted server collects users' locations at each time point in Figure \ref{fig:corr_example}(a) and continuously publishes aggregate (i.e., the counts of people at each location) in Figure\ref{fig:corr_example}(c) with differential privacy.
Our goal is to achieve $ \epsilon $-DP at each time point $ t $ (event-level $ \epsilon $-DP\cite{dwork_differential_2010-2}
\cite{dwork_differential_2010}) where $ t\in[1,T] $.
Suppose that each user appears at only one location at each time point.
According to the Laplace mechanism\cite{dwork_calibrating_2006}, adding $ Lap(1/\epsilon)$ noise\footnote{$Lap({b})$ denotes a Laplace distribution with variance $ 2b^2 $.} to perturb each count in Figure\ref{fig:corr_example}(c) can achieve $ \epsilon $-DP at each time point. 
However,  the expected privacy guarantee may decay due to temporal correlations as follows. 
Using auxiliary information, such as road networks, an attacker may know users' mobility patterns, such as ``always arriving at $ loc_5 $ after visiting $ loc_4 $''  (shown in Figure \ref{fig:corr_example}(b)), leading to the patterns illustrated in solid lines of Figure \ref{fig:corr_example}(c). 
The temporal correlation due to this road network can be represented as $ \Pr(l^{t}=loc_5|l^{t-1}=loc_4) =1$ where $ l^t $ is the location of a user at time $ t $.
That is, given the previous counts of $ loc_4 $, an attacker can derive the current count of $ loc_5 $.
Consequently, because an adversary can perform inference due to such correlations between the two consecutive time points (i.e., as if the same count is released two times),  adding $ Lap(1/\epsilon) $ noise to each count guarantees $ 2\epsilon $-DP at the time point.
Furthermore, let us consider an extreme case of temporal correlation (e.g., a terrible traffic congestion) $ \Pr(l^{t}=loc_4|l^{t-1}=loc_4) = \Pr(l^{t}=loc_5|l^{t-1}=loc_5) =1$ (i.e., the counts of $ loc_4 $ and $ loc_5 $  will not change over time). 
Then, adding $ Lap(1/\epsilon) $ noise to each count guarantees $T\epsilon $-DP at time point $ T $.

\end{example}

It is reasonable to consider that adversaries may obtain the temporal  correlations, which commonly exist in our real life and are easily acquired from public information or historical data.
In addition to road networks,
there are countless factors that may cause temporal correlations such as the common patterns characterizing human activities\cite{gambs_next_2012} and weather conditions\cite{horanont_weather_2013}. 

Few studies in the literature investigated such potential privacy loss of event-level $ \epsilon $-DP  under temporal correlations as shown in Example \ref{example}.
A direct method (without finely utilizing the \textit{probability} of correlation) involves amplifying the perturbation in terms of \textit{group differential privacy}\cite{chen_correlated_2014}\cite{dwork_calibrating_2006}, i.e., protecting the correlated data as a group.
In Example \ref{example}, for temporal correlation $ \Pr(l^{t}=loc_5|l^{t-1}=loc_4) =1$, we can protect the counts of $ loc_4 $ at time $ t -1$ and $ loc_5 $ at time $ t $ in a group (the sensitivity becomes 2) by increasing the scale of the perturbation to $ Lap(2/\epsilon) $ at each time point; 
for temporal correlation $ \Pr(l^{t}=loc_i|l^{t-1}=loc_i) =1$, in order to guarantee $ \epsilon $-DP at time $ T $, we need to add $ Lap(T/\epsilon) $ noise at each time point because the privacy leakage accumulates over time.
However, this technique is not suitable for \textit{probabilistic} correlations to finely prevent privacy leakage and may over-perturb the data as a result.  
For example, regardless of whether $ \Pr(l^{t}=loc_i|l^{t-1}=loc_i) $ is 1 or 0.1, it always protects the correlated data in a bundle.

Although a few studies investigated the issue of differential privacy under \textit{probabilistic} correlations, they are not applicable for \textit{continuous data release} because of the different problem settings.
The following two works focused on one-shot data release and different types of correlations.
Yang et al. \cite{yang_bayesian_2015} proposed Bayesian differential privacy (BDP),  which measures the privacy leakage under probabilistic correlations between tuples, modeled by a Gaussian Markov Random Field without taking time factor into account.
Liu et al. \cite{changchang_liu_dependence_2016} proposed \textit{dependent differential privacy} by introducing two parameters of \textit{dependence size} and \textit{probabilistic dependence relationship} between tuples.
However, it is not clear whether we can specify them for temporally correlated data.
Another line of work\cite{shokri_quantifying_2011}\cite{theodorakopoulos_prolonging_2014}\cite{xiao_protecting_2015} has investigated adversaries with knowledge of temporal correlations. 
They focused on designing new mechanisms for protecting a single user's location privacy extending DP, whereas we attempt to quantify the potential privacy loss of a traditional DP mechanism in the context of continuous aggregate release.

We call the adversary considered in traditional DP with additional knowledge of probabilistic temporal correlations \textit{adversary$ _\mathcal{T} $}.
Rigorously quantifying and bounding the privacy leakage against adversary$ _\mathcal{T} $ remains a challenge.
Therefore, our goal is to solve the following problems in this paper:
\begin{itemize}
\item How do we formalize \textit{adversary$ _\mathcal{T} $} and define the privacy loss against adversary$_\mathcal{T} $? (Section \ref{sec:privacy_analysis})
\item How do we calculate the privacy loss against adversary$_\mathcal{T} $?  (Section \ref{sec:calculate_tpl})
\item How do we bound the privacy loss against adversary$_\mathcal{T} $?  (Section \ref{sec:release_mechanism})
\end{itemize}

\subsection{Contributions}
In this work, for the first time, we quantify and bound the privacy leakage of a DP mechanism due to temporal correlations.
Our contributions are summarized as follows.

First, we rigorously define \textit{adversary$ _\mathcal{T} $} with temporal correlations that are described by the commonly used Markov model.
The temporal correlations include \textit{backward} and \textit{forward} correlations, i.e., $ \Pr(l_i^{t-1}|l_i^{t}) $ and $ \Pr(l_i^{t}|l_i^{t-1}) $ where $ l_i^t $ denotes the value (e.g., location) of user $ i $ at time $ t $.
We then define \textit{Temporal Privacy Leakage (TPL)}  as the privacy loss of a DP mechanism at time $ t $ against adversary$_\mathcal{T} $.
TPL includes two parts: \textit{Backward Privacy Leakage (BPL)} and \textit{Forward Privacy leakage (FPL)} due to the existence of backward and forward temporal correlations.
Our analysis shows that BPL may accumulates from previous privacy leakage and FPL increases with future release.
Intuitively, BPL at time $ t $ is affected by previously released data and FPL at time $ t $ will be affected by future releases.
We define $ \alpha $-\textit{differential privacy under temporal correlation}, denoted as $ \alpha$-$DP_\mathcal{T}$, to formalize the privacy guarantee of a DP mechanism against adversary$_\mathcal{T} $, i.e., the temporal privacy leakage should be bounded in $ \alpha $.
We prove a new form of sequential composition theorem for {\small $ \alpha$-$DP_\mathcal{T}$} (different from the traditional sequential composition\cite{mcsherry_privacy_2009} for $ \epsilon $-DP).

Second, we efficiently calculate the temporal privacy leakage under given backward and forward temporal correlations.
We transform the calculation of temporal privacy leakage in finding an optimal solution of a \textit{linear-fractional programming} problem.
This type of problem can be solved using a simplex algorithm in exponential time.
By exploiting the constraints, we propose a polynomial-time algorithm to finely quantify the temporal privacy leakage.

Third, we design private data release algorithms that can be used to convert a traditional DP mechanism into one satisfying $ \alpha $-DP$_\mathcal{T} $.
A challenge is that the temporal privacy leakage may increase over time so that $ \alpha $-DP$_\mathcal{T} $ is hard to achieve when the length of release time $ T $ is unknown.
In our first solution, we prove that the supremum of  temporal privacy leakage may exist in some cases, and in these cases, we allocate appropriate privacy budgets to make sure the increased temporal privacy leakage will never be greater than $ \alpha $, no matter how long the $ T $ is.
However,  when $ T $ is
too short for the accumulation of temporal privacy leakage to result in a significant increase, 
we may over-perturb the data.
The second solution is to exactly achieve $ \alpha $-DP$_\mathcal{T} $ at each time point by finely calculating the temporal privacy leakage.

Finally, experiments with synthetic data confirm the efficiency and effectiveness of our privacy leakage quantification algorithm. 
We also demonstrate the impact of different degree of temporal correlations on privacy leakage.

\section{Preliminaries}

\subsection{Differential Privacy}
Differential privacy\cite{dwork_differential_2006} is a formal definition of data privacy.
Let $ D $  be a database and $ D' $ be a copy of $ D $ that is different in any one tuple.
$ D $ and $ D' $ are \textit{neighboring databases}. 
A differentially private output from $ D $ or $ D' $  should exhibit little difference.

\begin{definition}[$ \epsilon $-DP]
\label{def:dp}
$\mathcal{M}$ is a randomized mechanism that takes as input $ D $ and outputs $ \bm{r} $, i.e., {\small $ \mathcal{M}(D)=\bm{r}$}.
$\mathcal{M}$ satisfies $ \epsilon $-differential privacy ($ \epsilon $-DP) if the following inequality is true for any pair of neighboring databases $ D, D'$ and all possible outputs $ \bm{r} $.
\vspace{-5pt}
\begin{myAlignSS}
\small
\log\frac{\Pr(\bm{r}|D)}{\Pr(\bm{r}|D')} \leq \epsilon.  \label{eq:dp}
\end{myAlignSS}
 \end{definition}
 
 \vspace{-12pt}
The parameter {\small $ \epsilon $}, called the \textit{privacy budget}, represents the degree of privacy offered.
Intuitively, a lower value of {\small $ \epsilon $}  implies stronger privacy guarantee and a larger perturbation noise, and a higher value of  {\small $ \epsilon $} implies a weaker privacy guarantee while possibly achieving higher accuracy.

A commonly used method to achieve $ \epsilon $-DP is the Laplace mechanism, which adds random noise drawn from a calibrated Laplace distribution into the aggregates to be published.

\begin{theorem}[Laplace Mechanism]
Let {\footnotesize $ Q:D\rightarrow \mathbb{R} $} be a statistical query on database $ D $.
The sensitivity of {\footnotesize $ Q $}  is defined as the maximum $ L_1 $ norm between {\footnotesize $ Q(D) $} and {\footnotesize $ Q(D') $}, i.e., {\footnotesize $ \Delta=\max_{D,D'}||Q(D)-Q(D')||_1$}.
We can achieve $ \epsilon $-DP by adding Laplace noise with scale {\footnotesize $ \Delta/\epsilon $}, i.e., {\footnotesize $ Lap(\Delta/\epsilon) $}.
\end{theorem}

\subsection{Privacy Leakage}
\label{subsec:pl}
Let us first discuss the adversaries tolerated by differential privacy, and then formalize privacy leakage w.r.t. such adversaries.
Differential privacy is able to protect against the attackers who even have knowledge of all users' data in the database except the one of the targeted victim\cite{dwork_algorithmic_2013}.
Let $ i \in \bm{U}$ be a user in the database $ D $.
Let $ A_i$ be an adversary who targets user $ i $ and has knowledge of {\small $ D_\mathcal{K} =D-\{l_i\}$} where {\small $ l_i \in [loc_1,\ldots,loc_n]$} denotes the data of user $ i $.
The adversary $ A_i $ observes the private output $ \bm{r} $ and attempts to guess whether the possible value of user $ i $ is  {\small $ loc_j$}  or {\small $ loc_k$}  where {\small $  loc_j, loc_k \in [loc_1,\ldots,loc_n]$}.
%
%
We define the privacy leakage of a DP mechanism as follows.
\begin{definition}[Privacy Leakage of a DP mechanism against $ A_i $]
\label{def:pl}
Let $ \bm{U} $ be a set of users in the database.
Let $ A_i$ be an adversary who targets user $ i $ and knows all the tuples in the database except the one of user $ i $.
The \textit{privacy leakage} of a differentially private mechanism $\mathcal{M}$ against one $ A_i $ and all $ A_i, i\in\bm{U}$ are defined, respectively, as follows in which $ l_i $ and $ l_i' $ are two different possible values of user $ i $'s data.
\begin{myAlignSSS}
&PL_0(A_i,\mathcal{M}) \eqdef \sup_{\bm{r},l_i,l_i'} \log \frac{\Pr(\bm{r}|l_i,D_\mathcal{K})}{\Pr(\bm{r}|l_i',D_\mathcal{K})} \nonumber\\
&PL_0(\mathcal{M}) \eqdef \max_{\forall A_i,i\in\bm{U}} PL_0(A_i,\mathcal{M}) = \sup_{\bm{r},D,D'} \log  \frac{\Pr(\bm{r}|D)}{\Pr(\bm{r}|D')} \nonumber
\end{myAlignSSS}
\end{definition}
\vspace{-10pt}
In other words, the privacy budget of a DP mechanism can be considered as a metric of \textit{privacy leakage}.
The larger {\small $ \epsilon $}, the larger the privacy leakage.
Hence, we can say  that $ \mathcal{M} $ satisfies $ \epsilon $-DP if  {\small $ PL_0(\mathcal{M}) \leq \epsilon$}.
We note that a ${\epsilon}'$-DP mechanism automatically satisfies $ \epsilon $-DP if $  {\epsilon}' < {\epsilon}$. 
For convenience, in the following parts of this paper, when we say that $ \mathcal{M} $ satisfies $ \epsilon $-DP, we mean that the privacy leakage is equal to $ \epsilon $.

\vspace{-6pt}
\subsection{Problem Setting}
\label{sec:problem_setting}
We attempt to quantify  the potential privacy loss of a DP mechanism under temporal correlations in the context of \textit{continuous data release} (e.g., releasing private counts at each time as shown in Figure \ref{fig:corr_example}).
 Users in the database, denoted by $ \bm{U} $,  are generating data continuously.
Let {\small $ \bm{loc}=\{loc_1,\ldots,loc_n\} $} be all possible values of user's data.
We denote the value of user $ i $ at time $ t $ by $ l_i^t $. 
A trusted server collects the data of each user into the database {\small $ D^t  =\{l_1^t,\ldots,l_{|\bm{U}|}^t\} $} at each time $t$ (e.g., the columns in Figure \ref{fig:corr_example}(a)).
A DP mechanism $ \mathcal{M}^t $ releases differentially private output $ \bm{r}^t$ independently at each time $ t $.
Our goal is to quantify and bound the potential privacy loss of $ \mathcal{M}^t $ against adversaries with knowledge of temporal correlations. 
We summarize the notations used in this paper in Table \ref{tb:notation}.
We note that while we use location data in Example \ref{example}, the problem setting is general for temporally correlated data.

\vspace{-2pt}
Our problem setting is identical to \textit{differential privacy under continual observation} in the literature
\cite{acs_case_2014}
\cite{bolot_private_2013}
\cite{chan_private_2011}
\cite{dwork_differential_2010-2}
\cite{dwork_differential_2010}
\cite{fan_fast:_2013}
\cite{kellaris_differentially_2014}
\cite{xiao_dpcube:_2014}.
In contrast to ``one-shot'' data release over a static database, the adversaries can observe multiple differentially private outputs, i.e., $ \bm{r}^1,\ldots,\bm{r}^t$.
There are typically two different privacy goals in the context of continuous data release: \textit{event-level} and \textit{user-level}
\cite{dwork_differential_2010-2}
\cite{dwork_differential_2010}. 
The former protects each user's single data point at time $ t $ (i.e., the neighboring databases are {\small $ D^t $} and {\small $ {D^t}' $}), whereas the latter protects the presence of a user with all her data on the timeline (i.e., the neighboring databases are {\small $\{ D^1 ,\ldots,D^t \}$} and {\small $\{ {D^1}',\ldots,{D^t}' \}$}).
In this work, we mainly study the privacy leakage at a single time point (event-level) under temporal correlations, and we also extend the discussion to user-level privacy by studying the composability of the privacy leakage.


\renewcommand{\arraystretch}{1.25}
\begin{table}[t]
\centering
\scriptsize
\caption{Summary of Notations.}
\begin{tabularx}{8.5cm}{|p{1cm}|X|}
\hline $ \bm{U}$ & The set of users in the database \\
\hline $ i$ & The $ i $-th user where $ i \in [1, |\bm{U}|] $\\
\hline $ \bm{loc}$ & Value domain $ \{loc_1,\ldots, loc_n\} $ of all user's data \\ 
\hline $ l_i^t, {l_i^t}' $&  The data of user $ i $ at time $ t $, $ l_i^t \in \bm{loc} $, $ l_i^t \neq {l_i^t}' $\\
\hline ${D}^t$ & The database at time $ t$, ${D}^t=\{l_1^t,\ldots,l_n^t\} $\\
\hline $ \mathcal{M}^t $ & Differentially private mechanism over $ D^t $\\
\hline $ \bm{r}^t $ & Differentially private output  at time $t$ \\ 
\hline $ A_i $ &   The adversary who targets user $i $, considered in traditional DP\\ 
\hline $ A_i^{\mathcal{T}} $ &  Adversary $ A_i $ with additional knowledge of temporal correlations  \\ 
\hline $P_i^B$ &  Transition matrix that represents {\scriptsize$ \Pr(l_i^{t-1}|l_i^t) $}, \newline i.e., backward temporal correlation, known to $ A_i^\mathcal{T} $ \\  
\hline $P_i^F$ &  Transition matrix that represents {\scriptsize$ \Pr(l_i^{t}|l_i^{t-1}) $}, \newline i.e., forward temporal correlation, known to $ A_i^\mathcal{T} $ \\  
\hline ${D}_\mathcal{K}^t$ & The subset of database $ D^t - \{l_i^t\} $, known to $A_i^\mathcal{T} $\\
\hline 
\end{tabularx}
\label{tb:notation} 
\end{table}

\vspace{-6pt}
\section{Analyzing Privacy Leakage}
\label{sec:privacy_analysis}
In the following, we first formalize adversary with temporal correlations in Section \ref{subsec:adv}.
We then define and analyze \textit{temporal privacy leakage} in Section \ref{subsec:tpl}.
We provide a new privacy notion of $ \alpha \text{-} DP_\mathcal{T}$ against temporal privacy leakage and prove its composability in Section \ref{subsec:comp}.
Finally, we make a few important observations in Section \ref{subsec:discussion}.

\vspace{-5pt}
\subsection{Adversay with Knowledge of Temporal Correlations}
\label{subsec:adv}
\vspace{-2pt}
\textbf{Markov Chain for Temporal Correlations.}
The \hyphenation{Mar-kov} Markov chain (MC) is extensively used in modeling user mobility profiles\cite{gambs_next_2012}\cite{mathew_predicting_2012}\cite{shokri_quantifying_2011}.
For a time-homogeneous first-order MC, a user's current value only depends on the previous one.
The parameter of the MC is the \textit{transition matrix}, which describes the probabilities for transition between values.
The sum of the probabilities in each row of the transition matrix is $ 1 $.
A concrete example of transition matrix and time-reversed one  for location data  is shown in Figure \ref{fig:mc}.
As shown in Figure \ref{fig:mc}(a), if user $ i $ is at $ loc_1 $ now (time $ t $); then, the probability of coming from $ loc_3 $  (time $ t-1 $) is $0.7$, namely, {\small $ \Pr( l_i^{t-1} = loc_3 | l_i^{t} =loc_1 ) = 0.7$}.
As shown in Figure \ref{fig:mc}(b), if user $ i $ was at $ loc_3 $ at the previous time $ t-1 $, then the probability of being at $ loc_1 $ now (time $ t $) is $0.6$; namely, {\small $ \Pr( l_i^{t} = loc_1 | l_i^{t-1} =loc_3 ) = 0.6$}.
We call the transition matrices in Figure \ref{fig:mc}(a) and (b) as backward temporal correlation and forward temporal correlation, respectively.

\begin{definition}[Temporal Correlations]
\label{def:temporal_corr}
The backward and forward temporal correlations between user $ i $'s data  $ l_i^{t-1} $ and $  l_i^{t} $ are described by  transition matrices {\small ${P}_i^B,{P}_i^F \in \mathbb{R}^{n \times n}$}, representing {\small $ \Pr( l_i^{t-1} | l_i^{t}) $} and {\small $ \Pr( l_i^{t} | l_i^{t-1}) $}, respectively.
\end{definition}

\vspace{-5pt}
It is reasonable to consider that the backward and/or forward temporal correlations could be acquired by adversaries.
For example, the adversaries can learn them from user's historical trajectories (or the reversed trajectories) by well studied methods such as 
Maximum Likelihood estimation (supervised) or Baum-Welch algorithm (unsupervised). 
Also, if the initial distribution of {\small $ l_i^1 $}  is known (i.e., {\small $ \Pr(l_i^1) $}), the backward temporal correlation (i.e., {\small $ \Pr(l_i^{t-1}|l_i^{t}) $}) can be derived from the forward temporal correlation (i.e., {\small $ \Pr(l_i^{t}|l_i^{t-1}) $}) by the following Bayesian inference. 
\vspace{-10pt}
\begin{myAlignSS}
\Pr(l^{t-1}_i | l^t_i) = \frac{\Pr(l^{t}_i | l^{t-1}_i)*\Pr(l^{t-1}_i)}{\sum\nolimits_{l^{t-1}_i}{\Pr(l^{t}_i | l^{t-1}_i)*\Pr(l^{t-1}_i)}} \nonumber
\end{myAlignSS}


\vspace{-10pt}
\noindent
Since estimating temporal correlations from data is beyond the scope of this work, we assume the adversaries' prior knowledge about temporal correlations is given in our framework.

We now define  an ``updated version'' of $ A_i $ (in Definition \ref{def:pl}) with knowledge of temporal correlations.

\begin{definition}[Adversary$ _\mathcal{T }$]
\label{def:adversary_t}
 Adversary$ _\mathcal{T }$ is a class of adversaries who have knowledge of (1)  all other users' data {\small $ D_\mathcal{K}^t $} at  each time $ t $ except the one of the targeted victim, i.e., {\small $ D_\mathcal{K}^t=D^t-\{l_i^t\} $}, and (2) backward and/or forward temporal correlations represented as transition matrices {\footnotesize$ P_i^B $} and {\footnotesize $ P_i^F $}.
We denote Adversary$ _\mathcal{T }$ who targets user $ i $ by {\footnotesize $ A_i^\mathcal{T}({P}_i^B, {P}_i^F)$}.
\end{definition}

\begin{figure}[t]
\centering
\includegraphics[scale=0.28]{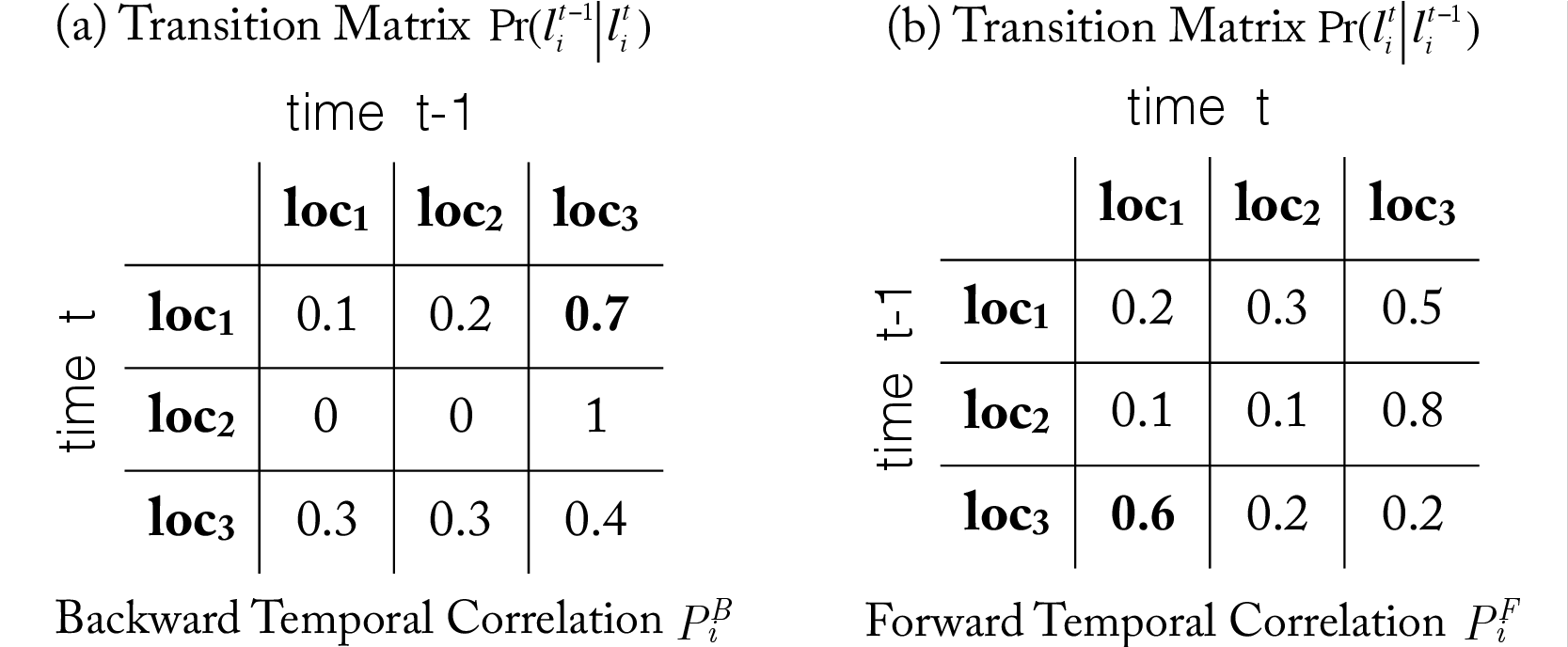} 
\caption{Examples of Temporal Correlations.}
\label{fig:mc}
\vspace{-8pt}
\end{figure}

\vspace{-5pt}
There are three types of adversary$ _\mathcal{T} $: 
(i) {\footnotesize $ A_i^\mathcal{T}({P}_i^B, \emptyset) $}, (ii) {\footnotesize $ A_i^\mathcal{T}(\emptyset, {P}_i^F)$},  (iii) {\footnotesize $ A_i^\mathcal{T}({{P}_i^B,{P}_i^F})$},
where $ \emptyset $ denotes that  the corresponding correlations are not known to the adversaries\footnote{The adversaries of types (i) and (ii) will not ``guess'' the missing correlations; otherwise, they fall under type (iii).}.
For simplicity, we denote types (i) and (ii) as {\footnotesize $ A_i^\mathcal{T}({P}_i^B) $} and  {\footnotesize $ A_i^\mathcal{T}({P}_i^F)$}, respectively.
We note that {\footnotesize $ A_i^\mathcal{T}(\emptyset, \emptyset) $} is the same as the traditional DP adversary $ A_i $ without any knowledge of temporal correlations. 

\begin{figure*}
\centering
\includegraphics[scale=0.33]{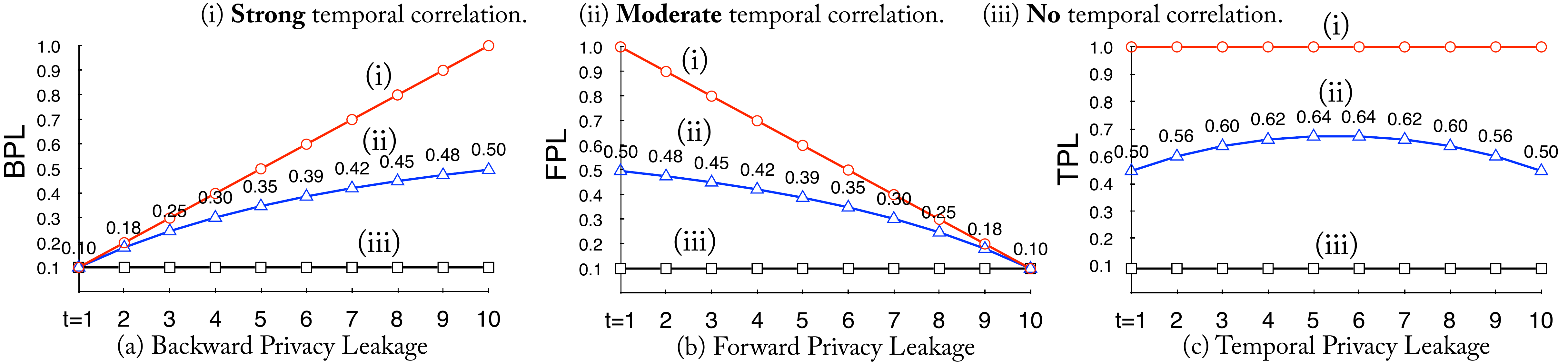} 
\caption{Example of Temporal Privacy Leakage of $ Lap(1/0.1) $ at each time point.}
\label{fig:tpl}
\vspace{-20pt}
\end{figure*}

\vspace{-5pt}
We now show what information $ A_i^\mathcal{T} $ can derive from the temporal correlations.
\begin{lemma}
\label{lem:tpl_b}
The adversary  {\footnotesize $ A_i^\mathcal{T} $} who has knowledge of $ P_i^B $ can derive {\footnotesize $\Pr(D^{t-1}|D^{t} )  = \Pr({l_i^{t-1}}| {l_i^{t}})  $}.
\end{lemma}


\begin{lemma}
\label{lem:tpl_f}
The adversary  {\small $ A_i^\mathcal{T} $} who has knowledge of $ P_i^F $ can derive {\small $\Pr(D^{t}|D^{t-1} )  = \Pr({l_i^{t}}| {l_i^{t-1}})  $}.
\end{lemma}

\vspace{-5pt}
\noindent
We omit the proofs of the lemmas due to space limitations.

\vspace{-4pt}
\subsection{Temporal Privacy Leakage}
\label{subsec:tpl} 
We now define the privacy leakage w.r.t. adversary$ _\mathcal{T} $.
For the convenience of analysis, let us assume the length of release time\footnote{In this paper, we do not need to know the length of release time in advance.} is $ T $.
The adversary {\small $ A_i^\mathcal{T} $} observes the differentially private outputs $ \bm{r}^1,\ldots,\bm{r}^t,\ldots,\bm{r}^T $ and attempts to infer the value of user $ i $'s data at time $ t $, namely $ l_i^t $.
Similar to Definition \ref{def:pl}, we define the privacy leakage in terms of event-level differential privacy in the context of continual data release as described in Section \ref{sec:problem_setting}.

\begin{definition}[Temporal Privacy Leakage, TPL]
\label{def:tpl}
Let {\small $ {D^t}' $} be a neighboring database of {\small $ D^t $}.
Let {\small $D_{\cal K}^t $} be the tuple knowledge of {\small $ A_i^\mathcal{T} $}.
We have {\footnotesize $ {D^t}'=D_{\cal K}^t  \cup \{{l_i^t}\} $} and {\footnotesize $ {D^t}'=D_{\cal K}^t  \cup \{{l_i^t}'\} $} where $l_i^t  $ and $ {l_i^t}' $ are two different values of user $ i $'s data at time $ t $.
Temporal Privacy Leakage (TPL)  of {\small $ \mathcal{M}^t $} w.r.t. a single {\small $ A_i^\mathcal{T} $}  and all {\small $ A_i^\mathcal{T}, i\in \bm{U}$}  are defined,  respectively, as follows. 
\begin{myAlignSSS}
{\textit{TPL}}&(A_i^\mathcal{T} ,\mathcal{M}^t) &\eqdef &
\sup_{\substack{ l_i^t, {l_i^t}', \bm{r}^1, \ldots, \bm{r}^T} } \log
{{\Pr ({\bm{r}^1}, \ldots, {\bm{r}^{T}}|{l_i^t,D_\mathcal{K}^t})} \over {\Pr ({\bm{r}^{1}}, \ldots ,{\bm{r}^{T}}|{l_i^t}',{D_\mathcal{K}^t})}}. \label{eq:tpl1} 
\\
{\textit{TPL}}&(\mathcal{M}^t) &\eqdef &\max_{\forall A_i^\mathcal{T}, i\in \bm{U}} {TPL}(A_i^\mathcal{T} ,\mathcal{M}^t) \label{eq:tpl2}
\\
&&=& \sup_{\substack{ D^t, {D^t}', \bm{r}^1, \ldots, \bm{r}^T} } \log
{{\Pr ({\bm{r}^1}, \ldots, {\bm{r}^{T}}|{D^t})} \over {\Pr ({\bm{r}^{1}}, \ldots ,{\bm{r}^{T}}|{D^t}')}}. \label{eq:tpl3}
\end{myAlignSSS}
\end{definition}

\vspace{-8pt}
We first analyze {\small $ {\textit{TPL}}(A_i^\mathcal{T} ,\mathcal{M}^t) $} (i.e., Equation \eqref{eq:tpl1}) because it is key to solve Equation \eqref{eq:tpl2} or \eqref{eq:tpl3}.
We can rewrite  {\small $ {\textit{TPL}}(A_i^\mathcal{T} ,\mathcal{M}^t) $} as follows because $ \bm{r}^1,\ldots,\bm{r}^T $ are published independently by differentially private mechanism {\small $ \mathcal{M}^1,\ldots,\mathcal{M}^T $}.
\begin{myAlignSSS}
&\text{Eqn.}\eqref{eq:tpl1}=\sup_{\substack{ l^t, {l_i^t}', \bm{r}^1, \ldots, \bm{r}^T} }
 \log \frac{\Pr(\bm{r}^1|{l_i^t,D_\mathcal{K}^t})}{\Pr({\bm{r}^1|{{l_i^t}',D_\mathcal{K}^t}})} *\dots * \frac{\Pr(\bm{r}^T|{l_i^t,D_\mathcal{K}^t})}{\Pr({\bm{r}^T|{{l_i^t}',D_\mathcal{K}^t}})} 
 \nonumber \\
&=
\underbrace{
\sup_{\substack{\bm{r}^1,...,\bm{r}^t,\\{l_i^t},{l_i^t}'}} 
\log \frac{\Pr(\bm{r}^1,...,\bm{r}^t|{l_i^t,D_\mathcal{K}^t})}{\Pr({\bm{r}^1,...,\bm{r}^t|{{l_i^t}',D_\mathcal{K}^t}})}
}_{\text{{\footnotesize backward privacy leakage}}}
+
\underbrace{
 \sup_{\substack{\bm{r}^t,...,\bm{r}^T,\\ {l_i^t},{l_i^t}'}} 
 \log \frac{\Pr(\bm{r}^t,...,\bm{r}^T|{l_i^t,D_\mathcal{K}^t})}{\Pr(\bm{r}^t,...,{\bm{r}^T|{{l_i^t}',D_\mathcal{K}^t}})}
}_{\text{{\footnotesize forward privacy leakage}}}  \nonumber\\
&-
\underbrace{
\sup_{\bm{r}^t,{l_i^t},{l_i^t}'} 
\log \frac{\Pr(\bm{r}^t|{l_i^t,D_\mathcal{K}^t})}{\Pr({\bm{r}^t|{{l_i^t}',D_\mathcal{K}^t}})}
}_{\textit{PL}_0(A_i^\mathcal{T},\mathcal{M}^t)}  \label{eq:tpl_expanded}
\end{myAlignSSS}

\vspace{-7pt}
It is clear that {\footnotesize $ \textit{PL}_0(A_i^\mathcal{T},\mathcal{M}^t)  = \textit{PL}_0(A_i,\mathcal{M}^t) $} because {\small $ \textit{PL}_0 $}  indicates the privacy leakage w.r.t. one output $ \bm{r} $ (refer to Definition \ref{def:pl}). 
As annotated in the above equation, we define backward and forward privacy leakage as follows.

\vspace{2pt}
\begin{definition}[Backward Privacy Leakage, BPL]
The privacy leakage of {\footnotesize $ \mathcal{M}^t $} caused by {\footnotesize $ \bm{r}^1,...,\bm{r}^t $} w.r.t. {\footnotesize $ A_i^\mathcal{T} $} is called backward privacy leakage, defined as follows.
\begin{myAlignSS}
{\textit{BPL}}&(A_i^\mathcal{T} ,\mathcal{M}^t) &\eqdef &
\sup_{\substack{ l_i^t, {l_i^t}', \bm{r}^1, \ldots, \bm{r}^t} } \log
{{\Pr ({\bm{r}^1}, \ldots, {\bm{r}^{t}}|{l_i^t,D_\mathcal{K}^t})} \over {\Pr ({\bm{r}^{1}}, \ldots ,{\bm{r}^{t}}|{l_i^t}',{D_\mathcal{K}^t})}}. \label{eq:bpl1} 
\\
{\textit{BPL}}&(\mathcal{M}^t) &\eqdef & \max_{\forall A_i^\mathcal{T},i \in \bm{U}} {\textit{BPL}}(A_i^\mathcal{T} ,\mathcal{M}^t) . \label{eq:bpl2} 
\end{myAlignSS}
\end{definition}

\vspace{-8pt}
\begin{definition}[Forward Privacy Leakage, FPL]
The privacy leakage of {\small $ \mathcal{M}^t $} caused by {\small $ \bm{r}^t,...,\bm{r}^T$} w.r.t. {\small $ A_i^\mathcal{T} $} is called forward privacy leakage, defined by follows.
\vspace{-3pt}
\begin{myAlignSS}
{\textit{FPL}}&(A_i^\mathcal{T} ,\mathcal{M}^t) &\eqdef &
\sup_{\substack{ l_i^t, {l_i^t}', \bm{r}^t, \ldots, \bm{r}^T} } \log
{{\Pr ({\bm{r}^t}, \ldots, {\bm{r}^{T}}|{l_i^t,D_\mathcal{K}^t})} \over {\Pr ({\bm{r}^{t}}, \ldots ,{\bm{r}^{T}}|{l_i^t}',{D_\mathcal{K}^t})}}. \label{eq:fpl1} 
\\
{\textit{FPL}}&(\mathcal{M}^t) &\eqdef & \max_{\forall A_i^\mathcal{T},i \in \bm{U}} {\textit{FPL}}(A_i^\mathcal{T} ,\mathcal{M}^t) . \label{eq:fpl2} 
\end{myAlignSS}
\end{definition}
\vspace{-8pt}
\noindent
By substituting Equation \eqref{eq:bpl1} and \eqref{eq:fpl1} into \eqref{eq:tpl_expanded}, we have
\begin{myAlignSSS}
\textit{TPL}(A_i^\mathcal{T} ,\mathcal{M}^t)
=
\textit{BPL}(A_i^\mathcal{T} ,\mathcal{M}^t)
+
\textit{FPL}(A_i^\mathcal{T} ,\mathcal{M}^t)
-
\textit{PL}_0( A_i^\mathcal{T} ,\mathcal{M}^t). \label{eq:tpl_comp1}
\end{myAlignSSS}

\vspace{-8pt}
\noindent
Similarly, by expanding Equation \eqref{eq:tpl3} to one resembling Equation \eqref{eq:tpl_expanded} and combining it with  Equation \eqref{eq:bpl2} and \eqref{eq:fpl2}, we have

\begin{myAlignS}
\textit{TPL}(\mathcal{M}^t)
=
\textit{BPL}(\mathcal{M}^t)
+
\textit{FPL}(\mathcal{M}^t)
-
\textit{PL}_0(\mathcal{M}^t). \label{eq:tpl_comp2}
\end{myAlignS}

\vspace{-15pt}
\noindent
Intuitively, BPL and FPL are the privacy leakage w.r.t. the adversaries {\footnotesize $ A_i^{\mathcal{T}}(P_i^B)$ } and {\footnotesize $ A_i^{\mathcal{T}}(P_i^F)$ }, respectively.
TPL is the privacy leakage w.r.t. {\footnotesize $ A_i^{\mathcal{T}}(P_i^B, P_i^F)$}.
In Equation \eqref{eq:tpl_comp2}, we need to minus {\footnotesize $ PL_0(\mathcal{M}^t) $} because it is counted in both BPL and FPL.
We will show more details in the following analysis.

\textbf{BPL over time.}
For BPL, we first expand and simplify Equation \eqref{eq:bpl1} by Bayesian theorem and Lemma \ref{lem:tpl_b},  {\small \textit{BPL}$(A_i^\mathcal{T}, \mathcal{M}^t) $} is equal to
\begin{myAlignSSS}
&
\sup_{\substack{ l_i^t,{l_i^t}',  \\ \bm{r}^1,\ldots,\bm{r}^{t-1}} }  \log
{
{\sum_{l_i^{t-1}} \Pr ({\bm{r}^1}, \ldots ,{\bm{r}^{t-1}}|{l_i^{t-1}},D_\mathcal{K}^{t-1}) \Pr(l_i^{t-1}|{l_i^{t}})   }
\over 
{\sum_{{l_i^{t-1'}}} 
\underbrace{ \Pr ({\bm{r}^1},\ldots, {\bm{r}^{t-1}}|{l_i^{t-1}}',D_\mathcal{K}^{t-1})}_{ {\text{(i) } {\textit{BPL}}}(A_i^\mathcal{T}, \mathcal{M}^{t-1})} 
\underbrace{\Pr({l_i^{t-1'}}|{l_i^{t}}')}_{\text{(ii) } P_i^B} } 
} \nonumber \\
&+ \sup_{\substack{ l_i^t,{l_i^t}', \bm{r}^t} }  \log \frac{\Pr({\bm{r}^{t}}|{l_i^{t}},D_\mathcal{K}^{t})}{\underbrace{\Pr({\bm{r}^{t}}|{l_i^{t}}',D_\mathcal{K}^{t})}_{\text{(iii) } \textit{PL}_0(A_i^\mathcal{T},\mathcal{M}^{t})}}.  \label{eq:bpl_cal}
\end{myAlignSSS}

\vspace{-5pt}
\noindent
We now discuss the three annotated terms in the above equation.
The first term indicates BPL at the previous time {\small $ t-1 $}, the second term is the backward temporal correlation determined by $ P_i^B $, and the third term is equal to the privacy leakage w.r.t. adversaries in traditional DP (see Definition \ref{def:pl}).
Hence, BPL at time $ t $ depends on (i) BPL at time $ t-1 $, (ii) the backward temporal correlations, and (iii) the (traditional) privacy leakage of {\small $ \mathcal{M}^t $} (which is related to the privacy budget allocated to {\small $ \mathcal{M}^t $}).
By Equation \eqref{eq:bpl_cal}, we know that 
if {\small $ t=1 $},  {\small $  \textit{BPL}(A_i^\mathcal{T},\mathcal{M}^1 ) = \textit{PL}_0(A_i,\mathcal{M}^1)$};
if {\small $ t>1 $}, we have the following, where {\small $ \mathcal{L}^B(\cdot) $} is a \textit{backward temporal privacy loss function} for calculating the accumulated privacy loss.

\begin{myAlignS}
\textit{BPL}(A_i^\mathcal{T}, \mathcal{M}^t ) = \mathcal{L}^B \big( \textit{BPL}(A_i^\mathcal{T}, \mathcal{M}^{t-1}) \big)+\textit{PL}_0(A_i, \mathcal{M}^{t}) \label{eq:bpl_f}
\end{myAlignS}

\vspace{-15pt}
\noindent Equation \eqref{eq:bpl_f} reveals that the {\small BPL} is calculated recursively and may \textit{accumulate over time}, as shown in Example \ref{example:bpl} (Fig.\ref{fig:tpl}(a)).

\begin{example}[\textbf{BPL due to previous releases}]
\label{example:bpl}
Suppose that a DP mechanism {\small $ \mathcal{M}^t$} satisfies {\small $ \textit{PL}_0(\mathcal{M}^t)=0.1 $} for each time $ t\in[1,T] $, i.e., 0.1-DP at each time point.
We now discuss BPL at each time point w.r.t. {\footnotesize $ A_i^\mathcal{T} $} with knowledge of backward temporal correlations {\footnotesize $P_{i}^B$}.
In an extreme case, if {\footnotesize $ P_i^B  $} indicates the strongest correlation, say,  {\footnotesize $ P_i^B=\big(\begin{smallmatrix} 1&0\\ 0&1\end{smallmatrix} \big)$}, then, at time {\small $t $}, {\footnotesize $ A_i^\mathcal{T}$} 
knows {\footnotesize $ l_i^{t} =l_i^{t-1}=\cdots=l_i^1$}, i.e., {\footnotesize $ D^{t} =D^{t-1}=\cdots=D^1$} because of  {\footnotesize $ D^t=\{l_i^t\} \cup D_\mathcal{K}^t $} for any {\footnotesize $ t \in [1,T] $}.
Hence, the continuous data release {\footnotesize $ \bm{r}^1,\ldots,\bm{r}^t $} is equivalent to releasing the \textit{same} database multiple times; the privacy leakage at each time point will accumulate from previous time points and increase linearly (Figure \ref{fig:tpl}(a)(i)).
In another extreme case, if there is no backward temporal correlation that is known to {\footnotesize $ A_i^\mathcal{T} $} (e.g., for the $ A_i $ in Definition \ref{def:pl} or {\footnotesize $ A_i^\mathcal{T}(P_i^F) $}\footnote{In this case, given the current $ l_i^t $ and $ P_i^F $, i.e., $ \Pr(l_i^{t}|l_i^{t-1}) $, the adversary cannot derive  {\footnotesize $ l_i^{t} =l_i^{t-1}=\cdots=l_i^1$}.}),
the backward privacy leakage at each time point is {\footnotesize $ \textit{PL}_0(\mathcal{M}^t) $}, as shown in Figure \ref{fig:tpl}(a)(iii).
Figure \ref{fig:tpl}(a)(ii) depicts the backward privacy leakage caused by {\small $ P_i^B=\big(\begin{smallmatrix} 0.8&0.2\\ 0&1\end{smallmatrix} \big)$}, which can be finely quantified using our method (Algorithm \ref{algo:cal_bpl})  in Section \ref{sec:calculate_tpl}.
\end{example}

\textbf{FPL over time.}
For FPL,  similar to the analysis of BPL,  we expand and simplify Equation \eqref{eq:bpl1} by Bayesian theorem and Lemma \ref{lem:tpl_f},  {\small \textit{FPL}$(A_i^\mathcal{T}, \mathcal{M}^t) $} is equal to

\begin{myAlignSSS}
&
\sup_{\substack{ l_i^t,{l_i^t}',  \\ \bm{r}^{t+1},\ldots,\bm{r}^{T}} }  \log
{
{\sum_{l_i^{t+1}} \Pr (\bm{r}^{t+1}, \ldots ,{\bm{r}^{T}}|{l_i^{t+1}},D_\mathcal{K}^{t+1}) \Pr(l_i^{t+1}|{l_i^{t}})   }
\over 
{\sum_{{l_i^{t+1'}}} 
\underbrace{ \Pr ({\bm{r}^{t+1}},\ldots, {\bm{r}^{T}}|{l_i^{t+1}}',D_\mathcal{K}^{t+1})}_{ {\text{(i) } {\textit{FPL}}}(A_i^\mathcal{T}, \mathcal{M}^{t+1})} 
\underbrace{\Pr({l_i^{t+1'}}|{l_i^{t}}')}_{\text{(ii) } P_i^F} } 
} \nonumber \\
&+ \sup_{\substack{ l_i^t,{l_i^t}', \bm{r}^t} }  \log \frac{\Pr({\bm{r}^{t}}|{l_i^{t}},D_\mathcal{K}^{t})}{\underbrace{\Pr({\bm{r}^{t}}|{l_i^{t}}',D_\mathcal{K}^{t})}_{\text{(iii) } \textit{PL}_0(A_i^\mathcal{T},\mathcal{M}^{t})}}.  \label{eq:fpl_cal}
\end{myAlignSSS}
\vspace{-5pt}
\noindent
By Equation \eqref{eq:fpl_cal}, we know that 
if {\small $ t=T $},  {\footnotesize $  \textit{FPL}(A_i^\mathcal{T},\mathcal{M}^T ) = \textit{PL}_0(A_i,\mathcal{M}^T)$};
if {\small $ t<T $}, we have the following, where {\small $ \mathcal{L}^F(\cdot) $} is a \textit{forward temporal privacy loss function} for calculating the increased privacy loss due to FPL at the next time.

\begin{myAlignS}
\textit{FPL}(A_i^\mathcal{T}, \mathcal{M}^t ) = \mathcal{L}^F \big( \textit{FPL}(A_i^\mathcal{T}, \mathcal{M}^{t+1}) \big)+\textit{PL}_0(A_i, \mathcal{M}^{t}) \label{eq:fpl_f}
\end{myAlignS}

\vspace{-8pt}
\noindent Equation \eqref{eq:fpl_f} reveals that FPL is calculated recursively and may \textit{increase over time}, as shown in Example \ref{example:fpl} (Fig.\ref{fig:tpl}(b)).

\begin{example}[\textbf{FPL due to future releases}]
\label{example:fpl}
Considering the same setting in Example \ref{example:bpl}, we now discuss FPL at each time point w.r.t. {\footnotesize $ A_i^\mathcal{T} $} with knowledge of forward temporal correlations {\footnotesize $P_{i}^F$}.
In an extreme case, if {\footnotesize $ P_i^F  $} indicates the strongest correlation, say,  {\footnotesize $ P_i^F=\big(\begin{smallmatrix} 1&0\\ 0&1\end{smallmatrix} \big)$}, then, at time {\small $t $}, {\footnotesize $ A_i^\mathcal{T}$} 
knows {\footnotesize $ l_i^{t} =l_i^{t+1}=\cdots=l_i^T$}, i.e., {\footnotesize $ D^{t} =D^{t+1}=\cdots=D^T$} because of  {\footnotesize $ D^t=\{l_i^t\} \cup D_\mathcal{K}^t $} for any {\footnotesize $ t \in [1,T] $}.
Hence, the continuous data release {\footnotesize $ \bm{r}^t,\ldots,\bm{r}^T $} is equivalent to releasing the same database multiple  times; 
the privacy leakage at time $ t $ will increase when every time new release (i.e., {\footnotesize $ \bm{r}^{t+1} $},{\footnotesize $ \bm{r}^{t+2} $},...) happens.
For example,
we see that contrary to BPL, the FPL at time 1 is the highest (due to future releases at time 1 to 10) while FPL at time 10 is the lowest (since there is no future release with respect to time 10 yet). 
 \ul{When $ \bm{r}^{11} $ is released, all FPL at time $ t \in [1,10] $ will be updated}.
In another extreme case, if there is no forward temporal correlation that is known to {\footnotesize $ A_i^\mathcal{T} $} (e.g., for the $ A_i $ in Definition \ref{def:pl} or {\footnotesize $ A_i^\mathcal{T}(P_i^B) $}),
then the forward privacy leakage at each time point is {\footnotesize $ \textit{PL}_0(\mathcal{M}^t) $}, as shown in Figure \ref{fig:tpl}(b)(iii).
Figure \ref{fig:tpl}(b)(ii) depicts the forward privacy leakage caused by {\small $ P_i^F=\big(\begin{smallmatrix} 0.8&0.2\\ 0&1\end{smallmatrix} \big)$}, which can be finely quantified using our method (Algorithm \ref{algo:cal_bpl}) in Section \ref{sec:calculate_tpl}.
\end{example}

\vspace{2pt}
\begin{remark}
\label{remark}
The extreme cases shown in Example \ref{example:bpl} and \ref{example:fpl}  are the upper and lower bound  of BPL and FPL.
Hence, the backward temporal privacy loss function {\small $ \mathcal{L}^B(\cdot) $} in Equation \eqref{eq:bpl_f} and the forward temporal privacy loss function {\small $ \mathcal{L}^F(\cdot) $} in Equation \eqref{eq:fpl_f} satisfy  {\small $ 0 *\cdot \leq \mathcal{L}^B(\cdot)  \leq  1* \cdot$}, where $ \cdot $ is  BPL at the previous time, and {\small $ 0 *\cdot \leq \mathcal{L}^F(\cdot)  \leq  1* \cdot$}, where $ \cdot $ is FPL at the next time, respectively.
\end{remark}

\vspace{-5pt}

From Example \ref{example:bpl} and \ref{example:fpl}, we know that: backward temporal correlation (i.e.,{\footnotesize $ P_i^B $}) does not affect FPL, and forward temporal correlation (i.e.,{\footnotesize $ P_i^F $}) does not affect BPL. 
In other words, adversary {\footnotesize $A_i^\mathcal{T}(P_i^B)  $} only causes BPL;  {\footnotesize $A_i^\mathcal{T}(P_i^F)  $} only causes FPL; while  {\footnotesize $A_i^\mathcal{T}(P_i^B,P_i^F)  $} poses a risk on both BPL and FPL.

Figure \ref{fig:tpl}(c) shows TPL, which is calculated using BPL and FPL (see Equation \eqref{eq:tpl_comp2}).
Given {\footnotesize $P_i^B $} and {\footnotesize $P_i^F  $}, finely quantifying TPL is a challenge.
We will design a novel algorithm to calculate them efficiently in Section \ref{sec:calculate_tpl}.

\vspace{-5pt}

\subsection{DP under Temporal Correlations and Its Composability}
\label{subsec:comp}
In this section, we define $ \alpha $-DP$ _\mathcal{T}$  to provide a privacy guarantee against temporal privacy leakage.
We prove its sequential composition theorem and discuss the connection between $ \alpha $-DP$ _\mathcal{T}$ and  $ \epsilon $-DP in terms of event-level/user-level privacy\cite{dwork_differential_2010-2}\cite{dwork_differential_2010} and w-event privacy\cite{kellaris_differentially_2014}.

\vspace{6pt}
\begin{definition}[$ \alpha \text{-} DP_\mathcal{T} $]
\label{def:DP}
For all user $ i $ in the database, if TPL of {\small $\mathcal{M}^t$} (see Definition \ref{def:tpl}) is less than or equal to $ \alpha $, we say that {\small $\mathcal{M}^t$} satisfies $ \alpha$-differential privacy under temporal correlation, denoted by {\small $ \alpha \text{-} DP_\mathcal{T} $}.
\end{definition}


\vspace{-5pt}
\noindent
{\small DP$_\mathcal{T} $} is an enhanced version of DP on temporal data.
If the data are temporally independent (i.e., for all user $ i $, both {\small $ P_i^B $} and {\small $ P_i^F $} are $ \emptyset $), an $ \epsilon $-DP  mechanism satisfies  $ \epsilon $-DP$_\mathcal{T} $.
If the data are temporally correlated (i.e., existing user $ i $ whose  {\small $ P_i^B $} or {\small $ P_i^F $} is not $ \emptyset $), an $ \epsilon $-DP  mechanism satisfies  $ {\alpha}$-DP$_\mathcal{T} $ where $ {\alpha}$ is the increased privacy leakage and can be quantified using our framework.

One may wonder, for a sequence of DP$_\mathcal{T} $ mechanisms on the timeline, what is the overall privacy guarantee.
Suppose that {\footnotesize $\mathcal{M}^t $}  satisfies $ \epsilon_t $-DP and poses risks of BPL and FPL as {\footnotesize $ \alpha_t^B $} and {\footnotesize $ \alpha_t^F $}, respectively.
That is, {\footnotesize $\mathcal{M}^t $} satisfies  {\footnotesize $ ( \alpha_t^B + \alpha_t^F - \epsilon_t )$}-DP$ _\mathcal{T} $ at time $ t $ according to Equation \eqref{eq:tpl_comp2}.
We formally define such overall privacy leakage based on Equation \eqref{eq:tpl3}.
\begin{definition}[TPL of a sequence of DP mechanisms]
\label{def:combined_pl}
The temporal privacy leakage of DP mechanisms {\footnotesize $ \{\mathcal{M}^t,\ldots,\mathcal{M}^{t+j}\} $} where $ j \geq 0 $ is defined as follows.
\begin{myAlignSSS}
\textit{TPL}\big( \{\mathcal{M}^t,\ldots,\mathcal{M}^{t+j}\} \big) \eqdef
\sup_{\substack{ D^t,...,D^{t+j},\\ {D^t}',...,{D^{t+j}}',\\ \bm{r}^1, \ldots, \bm{r}^T} } \log
{{\Pr ({\bm{r}^1}, \ldots, {\bm{r}^{T}}|{D^t,\ldots,D^{t+j}})} \over {\Pr ({\bm{r}^{1}}, \ldots ,{\bm{r}^{T}}|{D^t}',\ldots, {D^{t+j}}')}} \nonumber 
\end{myAlignSSS}
\end{definition}
\vspace{-10pt}
It is easy to see that, if {\footnotesize $ j=0 $}, it is event-level privacy; if {\footnotesize $ t=1 $} and {\footnotesize $ j=|T-1| $}, it is user-level privacy.

\begin{theorem}[Composition under Temporal Correlations]
\label{thm:composition}
A sequence of DP mechanism {\footnotesize $\{\mathcal{M}^t,\ldots, \mathcal{M}^{t+j} \}$} satisfies 
\begin{myAlignSS}
\begin{cases}
{\scriptsize (\alpha_{t}^B + \alpha_{t+1}^F )}\textit{-DP}_\mathcal{T}  & j=1 \\
{\scriptsize \big(\alpha_{t}^B + \alpha_{t+j}^F + \sum_{k=1}^{k=j-1}\epsilon_{t+k} \big)}\textit{-DP}_\mathcal{T}  &   j \geq 2 
\end{cases}
\end{myAlignSS}
\end{theorem}

We omit the proofs of Theorems \ref{thm:composition} due to space limitations.

When {\small $ t=1 $} and {\small $ j=T-1 $} in Theorem \ref{thm:composition}, we have the following corollary because {\footnotesize $ \alpha_1^B=BPL(\mathcal{M}^1)=PL_0(\mathcal{M}^1) $} and {\footnotesize $ \alpha_T^F=FPL(\mathcal{M}^T)=PL_0(\mathcal{M}^T) $}.
\begin{corollary}
\label{col:user-level}
The temporal privacy leakage of a combined mechanism {\small $\{\mathcal{M}^1,\ldots, \mathcal{M}^{T} \}$} is  {\footnotesize $ \sum_{k=1}^{k=T}\epsilon_{k} $}. 
\end{corollary}
\noindent
It shows that temporal correlations do not affect the user-level privacy (i.e., protecting all the data on the timeline of each user), which is in line with the idea of group differential privacy: protecting all the correlated data in a bundle.

We now compare the privacy guarantee between DP and DP$ _\mathcal{T} $.
As we mentioned in Section \ref{sec:problem_setting}, there are typically two privacy notions in continuous data release: \textit{event-level} and \textit{user-level}
\cite{dwork_differential_2010-2}
\cite{dwork_differential_2010}.
Recently, $ w $-event privacy\cite{kellaris_differentially_2014} is proposed to merge the gap between event-level and user-level privacy.
It protects the data in any $ w $-length sliding window by utilizing the following sequential composition theorem of DP.

\vspace{2pt}
\begin{theorem}[Sequential composition on independent data\cite{mcsherry_privacy_2009}]
\label{thm:sequential_composition}
Suppose that {\footnotesize $ \mathcal{M}^t $} satisfies {\footnotesize $ \epsilon_t $}-DP for each {\footnotesize $ t\in[1,T] $}.
A combined mechanism {\footnotesize $\{\mathcal{M}^t,\ldots, \mathcal{M}^{t+j} \}$} satisfies {\footnotesize $ \sum_{k=1}^{k=j} \epsilon_{t+k} $}-DP.
\end{theorem}

Suppose that $ \mathcal{M}^t $ satisfies $ \epsilon $-DP for each $ t \in[1,T]$.
According to Theorem \ref{thm:sequential_composition}, it achieves $ T\epsilon $-DP on user-level and $ w\epsilon $-DP on $ w $-event level.
We compare the privacy guarantee on independent data and temporally correlated data as follows.

\renewcommand{\arraystretch}{1.25}
\begin{table}[h]
\centering
\footnotesize
\caption{The privacy guarantee of $ \epsilon $-DP mechanisms.}
\begin{tabularx}{8.5cm}{|p{2.8cm}|p{1.4cm}|X|}
\hline  
\diagbox[width=3.2cm]{{Privacy Notion}}{{Data}} &  independent &  temporally correlated  \\
\hline event-level{ \cite{dwork_differential_2010-2}\cite{dwork_differential_2010}}& $ \epsilon $-DP &  ${\alpha} $-DP$ _\mathcal{T}$ ({\small ${\alpha} \geq\epsilon$})  \\
\hline $ w $-event\cite{kellaris_differentially_2014} & $w \epsilon $-DP & see Theorem \ref{thm:composition}  \\
\hline user-level{ \cite{dwork_differential_2010-2}\cite{dwork_differential_2010}} & $ T\epsilon $-DP & $ T\epsilon $-DP$ _\mathcal{T} $ (by Corollary \ref{col:user-level})\\ 
\hline 
\end{tabularx}
\label{tb:diff} 
\end{table}

\vspace{5pt}
It reveals that temporal correlations may blur the boundary between event-level privacy and user-level privacy.
In an extreme case, the temporal privacy leakage of an $ \epsilon $-DP mechanism on event-level can be $ T\epsilon $, i.e., $ T\epsilon $-DP$_\mathcal{T}$.
Consider the examples shown in Figure \ref{fig:tpl}.
Under the strongest temporal correlations, 
{\small $ \mathcal{M}^{10} $} satisfies $ 1 $-DP$_\mathcal{T}$ on event-level and a combined mechanism {\footnotesize $ \{\mathcal{M}^1,\ldots,\mathcal{M}^{10} \}$} also satisfies $ 1 $-DP$_\mathcal{T}$ on user-level.
Essentially, it is because the adversaries may infer {\small $\{ D^1,\ldots,D^T \}$} (user-level) from $ D^t $ (event-level) using temporal correlations.
 
 

\subsection{Discussion}
\label{subsec:discussion}
We make a few important observations regarding our privacy analysis.

First, the temporal privacy leakage is defined in a personalized way. 
That is, the privacy leakage may be different for users with distinct temporal patterns (i.e., {\footnotesize $ P_i^B $} and {\footnotesize $ P_i^F $}).
We define the overall temporal privacy leakage as the maximum one for all users, so that $ \alpha $-DP$ _\mathcal{T} $ is compatible with the traditional $ \epsilon $-DP mechanism (using one parameter to represent the overall privacy level) and we can convert a traditional DP mechanism  to bound the temporal privacy leakage.
On the other hand, our definitions  also can be compatible with personalized differential privacy (PDP) mechanisms\cite{jorgensen_conservative_2015}, in which the personalized privacy budgets, i.e., a vector {\footnotesize $ [\epsilon_1, \ldots,\epsilon_n]$}, are allocated to each user.
In other words, we can convert a PDP mechanism to bound the temporal privacy leakage for each user.

Second, in this paper, we focus on the temporally correlated data and assume that the adversary has knowledge of temporal correlations modeled by Markov chain.
However, it is possible that the adversary has knowledge about more sophisticated temporal correlation model or
other types of correlations, such as user-user correlations modeled by Gaussian Markov Random Field in \cite{yang_bayesian_2015}.
Our contributions in this work can serve as primitives for quantifying the privacy risk under more advanced adversarial knowledge.

\section{Calculating Temporal Privacy Leakage}
\label{sec:calculate_tpl}
In this section, we design algorithms for computing backward privacy leakage (BPL) and forward privacy leakage (FPL).
We first show that both of them  can be transformed to the optimal solution of a \textit{linear-fractional programming} problem\cite{bajalinov_linear-fractional_2003}.
Traditionally, this type of problem can be solved by simplex algorithm\cite{dantzig_linear_1998} in exponential time.
By exploiting the constraints in this problem, we then design a method to solve it in polynomial time.

\subsection{Problem formulation}
According to the privacy analysis of BPL and FPL in Section \ref{subsec:tpl}, we need to solve the backward  and forward temporal privacy loss functions $ \mathcal{L}^B(\cdot) $ and $ \mathcal{L}^F (\cdot)$ in Equations \eqref{eq:bpl_f} and \eqref{eq:fpl_f}, respectively.
By observing the structure of the first term in Equations \eqref{eq:bpl_cal} and \eqref{eq:fpl_cal}, we can see that the calculations for recursive functions $ \mathcal{L}^B (\cdot)$ and $ \mathcal{L}^F (\cdot)$ are virtually in the same way.
They calculate the increment of the input values (the previous BPL or the next FPL) based on temporal correlations (backward or forward).
Although different degree of correlations result in different privacy loss functions, the methods for analyzing them are the same.

We now quantitatively analyze the temporal privacy leakage.
In the following, we demonstrate the calculation of $ \mathcal{L}^B (\cdot)$.
 The first term of Equation \eqref{eq:bpl_cal}, i.e., {\footnotesize $ \mathcal{L}^B(BPL(A_i^\mathcal{T},\mathcal{M}^{t-1})) $}  is as follows.
\begin{myAlignSSS}
\sup_{\substack{ l_i^t,{l_i^t}',  \\ \bm{r}^1,\ldots,\bm{r}^{t-1}} }  \log
{
{\sum_{l_i^{t-1}} \Pr ({\bm{r}^1}, \ldots ,{\bm{r}^{t-1}}|{l_i^{t-1}},D_\mathcal{K}^{t-1}) \Pr(l_i^{t-1}|{l_i^{t}})   }
\over 
{\sum_{{l_i^{t-1'}}} 
\underbrace{ \Pr ({\bm{r}^1},\ldots, {\bm{r}^{t-1}}|{l_i^{t-1}}',D_\mathcal{K}^{t-1})}_{ { {\textit{BPL}}}(A_i^\mathcal{T}, \mathcal{M}^{t-1})} 
\underbrace{\Pr({l_i^{t-1'}}|{l_i^{t}}')}_{P_i^B} } 
} \label{eq:bpl_cal_1st}
\end{myAlignSSS}

\vspace{-5pt}
We now simplify the notations in the above formula.
Let two arbitrary (different) rows in $ P_i^B $ be vectors {\footnotesize $ \bm{q}=(q_1,...,q_n)$} and {\footnotesize $\bm{d}=(d_1,...,d_{n}) $}.
For example, suppose that {\small $ \bm{q} $} is the first row in the transition matrix of Figure \ref{fig:mc}(b).  Then, the elements in {\small $ \bm{q} $} are:  {\footnotesize $q_1= \Pr(l_i^{t-1}=loc_1|l_i^t=loc_1) $}, {\footnotesize $q_2= \Pr(l_i^{t-1}=loc_2|l_i^t=loc_1) $}, {\footnotesize $q_3= \Pr(l_i^{t-1}=loc_3|l_i^t=loc_1) $}, etc.
Let {\footnotesize $ \bm{x}=(x_1,...,x_{n})^\mathrm{T}$} be a vector whose elements indicate {\footnotesize $\Pr ({\bm{r}^1},... ,{\bm{r}^{t-1}}|l_i^{t-1}, {D_{\mathcal{K}}^{t-1}})$}  with distinct values of {\small $ l_i^{t-1} \in \bm{loc} $}, e.g., {\footnotesize $ x_1 $} denotes {\footnotesize $\Pr ({\bm{r}^1},... ,{\bm{r}^{t-1}}|{l_i^{t-1}}=loc_1, {D_{\mathcal{K}}^{t-1}})$}.
We obtain the following by expanding {\small $ l_i^{t-1},l_i^{t-1'} \in \bm{loc} $} in \eqref{eq:bpl_cal_1st}.

\begin{myAlignSS}
\mathcal{L}^B \big( \textit{BPL}(A_i^\mathcal{T}, \mathcal{M}^{t-1})\big) 
&= \sup_{ \bm{q}, \bm{d} \in   P_i^B } \log \frac{q_1x_1+\cdots+q_n x_n}{d_1 x_1+\cdots+d_n x_n} \nonumber\\
&= \sup_{ \bm{q}, \bm{d} \in   P_i^B } \log  \frac{ {\boldsymbol{q}  \boldsymbol{x}}}{{\boldsymbol{d}   \boldsymbol{x}}} \nonumber
\end{myAlignSS}

\vspace{-8pt}
Next, we formalize the problem and constraints. 
Suppose that {\small $ \textit{BPL}(A_i^\mathcal{T}, \mathcal{M}^{t-1}) ={{\alpha}^B_{t-1}}$}.
According to the definition of BPL (as the supremum), for any {\footnotesize $ x_j,x_k \in \bm{x} $}, we have {\footnotesize $e^{-{{\alpha}^B_{t-1}}}  \leq \frac{x_j}{x_k} \leq e^{{{\alpha}^B_{t-1}}}$}.
Given $ \bm{x} $ as the variable vector and $ \bm{q},\bm{d} $ as the coefficient vectors,  {\footnotesize $ \mathcal{L}^B ({{\alpha}^B_{t-1}}) $}  is equal to the logarithm of the objective function \eqref{eq:lfp} in the following problem \eqref{eq:lfp}$ \sim $\eqref{eq:lfp_end2}.
\begin{myAlignSSS}
 \textnormal{maximize }  & \text{ }    \frac{ {\boldsymbol{q}  \boldsymbol{x}}}{{\boldsymbol{d}   \boldsymbol{x}}}  \label{eq:lfp}\\
 \textnormal{subject to  } 
 & \text{ } e^{-{{\alpha}^B_{t-1}}} \leq \frac{x_j}{x_k} \leq e^{{{\alpha}^B_{t-1}}},   \label{eq:lfp_end} \\
 &\text{ }  0 < {x_j}< 1 \text{ and }  0<x_k<1, \label{eq:lfp_end2}  \\
 & \text{ where } x_j,x_k \in \bm{x}, \text{ } j,k \in  [1,n]. \nonumber
\end{myAlignSSS}

\vspace{-18pt}
The above is a form of \textit{linear-fractional programming}\cite{bajalinov_linear-fractional_2003}, where the objective function is a ratio of two linear functions and the constraints are linear inequalities or equations.
A linear-fractional programming problem can be converted into a sequence of linear programming problems\cite{bajalinov_linear-fractional_2003} and then solved using the simplex algorithm\cite{dantzig_linear_1998} in time {\small $ O(2^{n}) $}.
When {\small $ n $} is large, the computation is time consuming.

\vspace{-5pt}
\textbf{Bounding the objective function by constraints.}
We now investigate a more efficient method to solve this problem by exploiting the structure of constraints.
From Inequalities \eqref{eq:lfp_end} and \eqref{eq:lfp_end2}, we know that the feasible region of the constraints are not empty and bounded; hence, an optimal solution exists.
By exploiting the constraints, we prove the following theorem, which enables the optimal solution to be found in time {\small $ O(n^2) $}.

\vspace{-5pt}
We define some notations that will be frequently used  in the following parts of this paper.
Suppose that the variable vector $ \bm{x} $ consists of two parts (subsets):  $ \bm{x}^+$ and $ \bm{x}^-$.
Let the corresponding coefficients vectors be {\small $ \bm{q}^+, \bm{d}^+ $} and {\small $ \bm{q}^-, \bm{d}^-$}.
Let {\small $ q=\sum{\bm{q}^+} $} and {\small $ d=\sum{\bm{d}^+} $}.
For example, suppose that  {\small $ \bm{x}^+=[x_1,x_3] $} and {\small $ \bm{x}^-=[x_2,x_4,x_5] $}. 
Then, we have {\small $ \bm{q}^+=[q_1,q_3] $}, {\small $ \bm{d}^+=[d_1,d_3] $}, {\small $ \bm{q}^-=[q_2,q_4,q_5] $}, and {\small $ \bm{d}^-=[d_2,d_4,q_5] $}. In this case, {\small $ q=q_1+q_3 $} and {\small $ d=d_1+d_3 $}.

\vspace{2pt}
\begin{theorem}
\label{thm:lfp}
If  the following Inequalities \eqref{eq:cond1} and \eqref{eq:cond2} are satisfied, the maximum value of the objective function in the problem \eqref{eq:lfp}$ \sim $\eqref{eq:lfp_end2} is {\footnotesize $ \frac{q(e^{{\alpha}_{t-1}^B} -1) + 1 }{d(e^{{\alpha}_{t-1}^B} - 1) + 1} $}.
\vspace{-5pt}
\begin{myAlignSSS}
&\frac{q_j}{d_j} > \frac{q(e^{{\alpha}_{t-1}^B} -1) + 1 }{d(e^{{\alpha}_{t-1}^B} - 1) + 1}, & \forall j\in[1,n] \text{ where } q_j \in \bm{q}^+, d_j \in \bm{d}^+ \label{eq:cond1}\\
&\frac{q_k}{d_k} \leq \frac{q(e^{{\alpha}_{t-1}^B} -1) + 1 }{d(e^{{\alpha}_{t-1}^B} - 1) + 1}, & \forall k\in[1,n] \text{ where } q_k \in \bm{q}^-, d_k \in \bm{d}^- \label{eq:cond2}
\end{myAlignSSS}
\vspace{-8pt}
\end{theorem}
\begin{proof}
See Appendix \ref{appex:thm_tpl}.
\end{proof}

\vspace{-5pt}
As we mentioned previously, the calculations of $ \mathcal{L}^B(\cdot) $ and  $ \mathcal{L}^F(\cdot) $ are identical.
Therefore, given a transition matrix $ P_i^B $ (or $ P_i^F $) and the previous BPL (or the next FPL), the increment of the backward (or forward) privacy loss is the maximum value in the above theorem for any two rows $ \bm{q} $ and $ \bm{d} $ in $ P_i^B $ (or $ P_i^F $).
We denote {\small $\textit{BPL}(A_i^\mathcal{T}, \mathcal{M}^{t-1}) $} and {\small $\textit{FPL}(A_i^\mathcal{T}, \mathcal{M}^{t+1}) $} by {\small $ {{\alpha}^B_{t-1}} $} and {\small $ {{\alpha}^F_{t+1}} $}, respectively.

\vspace{-8pt}
\begin{myAlignSSS}
&\mathcal{L}^B 
 \big( {{\alpha}_{t-1}^B} \big)  = \max_{ \bm{q}, \bm{d} \in   P_i^B } 
\log\frac{ q(e^{{\alpha}_{t-1}^B}-1) +1}{ d(e^{{\alpha}_{t-1}^B}-1) +1} \label{eq:lfp_sol_b}  \\
&\mathcal{L}^F
 \big( {{\alpha}_{t+1}^F} \big)  = \max_{ \bm{q}, \bm{d} \in   P_i^F } 
\log\frac{ q(e^{{\alpha}_{t+1}^F}-1) +1}{ d(e^{{\alpha}_{t+1}^F}-1) +1} \label{eq:lfp_sol_f} 
\end{myAlignSSS}

\vspace{-10pt}
It is easy to see that we can always find such {\small $ \bm{q}^+ $} and {\small $ \bm{d}^+ $} satisfying Inequalities \eqref{eq:cond1} and \eqref{eq:cond2}.
Further, we give the following corollary for finding {\small $ \bm{q}^+ $} and {\small $ \bm{d}^+ $}.
\begin{corollary}
\label{col:q_d}
If Inequalities \eqref{eq:cond1} and \eqref{eq:cond2} are satisfied, we have {\small $ q_j > d_j $} in which $ q_j \in \bm{q}^+ $ and  {\small $ d_j \in \bm{d}^+ $}.
\end{corollary}

\vspace{-3pt}

\vspace{-2pt}
Now, we simply examine the {\small $ \mathcal{L}^B (\cdot)$} and {\small $ \mathcal{L}^F (\cdot)$} in  Equations \eqref{eq:lfp_sol_b} and \eqref{eq:lfp_sol_f}. 
First, we have  {\small $ 0 \leq  \mathcal{L}^B ({{\alpha}^B_{t-1}})  $} and {\small $ 0 \leq  \mathcal{L}^F ({{\alpha}^B_{t-1}})  $} because of  {\small $q>d$}.
Second, when   $ q $ and $ d $ have the largest difference  (e.g., {\small$ \bm{q}=(1,0) $,$ \bm{d}=(0,1) $} and hence $q=1,d=0$), it follows that {\small $\mathcal{L}^B ({{\alpha}^B_{t-1}}) \leq {{\alpha}^B_{t-1}} $} and {\small $\mathcal{L}^F ({{\alpha}^F_{t+1}}) \leq {{\alpha}^F_{t+1}} $}.
Therefore, it is in accordance with Remark \ref{remark}.
The advantage of Equations \eqref{eq:lfp_sol_b} and \eqref{eq:lfp_sol_f} is being able to finely quantify BPL and FPL w.r.t. arbitrary $ P_i^B $ and $ P_i^F $.

\vspace{-8pt}
\subsection{Privacy Leakage Quantification Algorithm}
The next question is how do we find $ q $ and $ d $ (or {\small $ \bm{q}^+ $} and {\small $ \bm{d}^+ $}) that give the maximum objective function.
Inequalities \eqref{eq:cond1} and \eqref{eq:cond2} in Theorem \ref{thm:lfp} are sufficient conditions for obtaining such optimal value.
Corollary \ref{col:q_d} gives a necessary condition for satisfying Inequalities \eqref{eq:cond1} and \eqref{eq:cond2}.
Based on the above analysis, we design Algorithm \ref{algo:cal_bpl} for computing BPL or FPL .
\vspace{-2pt}
\begin{algorithm}[h]
\begin{footnotesize}
\SetKwRepeat{doWhile}{do}{while}
\SetKwComment{tcc}{//}{}
\DontPrintSemicolon 
\KwIn{$ P_i $  {\scriptsize ($  P_i^B$ or $  P_i^F$)};  $ \alpha $ ({\scriptsize ${ {{\alpha}^B_{t-1}}}$ or ${ {{\alpha}^F_{t+1}}}$}); $ \epsilon_t $ ({\scriptsize i.e., $\textit{PL}_0(\mathcal{M}^t) $}). }
\KwOut{BPL at time $ t $ or FPL at time $ t $}
     $ \mathcal{L} _i \longleftarrow 0$  \tcc*{{\scriptsize the value of Equation \eqref{eq:lfp_sol_b} or \eqref{eq:lfp_sol_f}}} 
    \ForEach{ two rows $ \bm{q} $, $ \bm{d} \in P_i $ }
    {  
            \ForEach(\tcc*[f]{\scriptsize Corollary \ref{col:q_d}}){$ q_j \in\bm{q},d_j \in\bm{d} $}    
                    { \lIf{$  q_j>d_j $} 
                    {add $ q_j $ to $ \bm{q}^+$; add $ d_j $ to $ \bm{d}^+  $ } 
            }
   $\textsf{\textit{update}} \longleftarrow false ;$\;
    \doWhile(\tcc*[f]{\scriptsize find $ \bm{q}^+ $, $ \bm{d}^+ $ by Theorem \ref{thm:lfp}} ) {\textsf{update}}{  
      {\scriptsize $q \longleftarrow \sum{\bm{q}^+}  $};  {\scriptsize $ d \longleftarrow \sum{\bm{d}^+} $} \tcc*{{\scriptsize update $ q $ and $ d $ }}
            \ForEach{$ q_j \in\bm{q}^+,d_j \in\bm{d}^+ $}
                    {  \tcc{{\scriptsize if it does not satisfy Inequality \eqref{eq:cond1}}} 
                    \lIf{{\scriptsize $  q_j/d_j \leq \big(q*(e^{{\alpha}}-1)+1 \big)/\big(d*(e^{{\alpha}}-1)+1 \big) $} \label{algo:iterate} \;} 
                          {\scriptsize    $ \bm{q}^+ \leftarrow \bm{q}^+ - q_j $;
                          $ \bm{d}^+ \leftarrow \bm{d}^+ - d_j $;
                          $\textsf{\textit{update}} \leftarrow true $;
                          \label{algo:findtpl:update}}  
            }
    }
            \lIf{\scriptsize $ \mathcal{L} _i    <  \log \frac{q*(e^{{\alpha}} -1) +1}{d*(e^{{\alpha}} -1) +1} $  }
            {\scriptsize $ \mathcal{L} _i    \longleftarrow \log \frac{q*(e^{{\alpha}} -1) +1}{d*(e^{{\alpha}} -1) +1} $}
    }
\Return{ {\scriptsize $  \mathcal{L} _i + \epsilon_t$}}   \tcc*{{\scriptsize by Equation \eqref{eq:bpl_f} or \eqref{eq:fpl_f}}}
\caption{Finding BPL or FPL}
\label{algo:cal_bpl}
\end{footnotesize}
\end{algorithm}

\vspace{-3pt}
\textbf{Computing BPL or FPL by solving the linear-fractional programming}.
According to the definition of BPL and FPL, we need to return the maximum privacy leakage (Line 12) w.r.t. any two rows in the given transition matrix (Line 2).
Lines 3$ \sim $11 are to solve one linear-fractional programming problem \eqref{eq:lfp}$ \sim $\eqref{eq:lfp_end2} w.r.t two specific rows in the transition matrix.
In Lines 3 and 4, we divide the variable vector $ \bm{x} $ into two parts according to Corollary \ref{col:q_d}, which gives the necessary condition for finding the maximum solution:
if  the coefficients {\footnotesize $ q_j \leq d_j $}, they are not in $ \bm{q}^+ $ and $ \bm{d}^+ $ that satisfy Inequalities \eqref{eq:cond1} and \eqref{eq:cond2}.
In other words, if {\footnotesize $ q_j > d_j $}, they are ``candidates'' in $ \bm{q}^+ $ and $ \bm{d}^+ $ that gives the maximum objective function.
In Lines 5$ \sim $11, we further check  $ \bm{q}^+ $ and $ \bm{d}^+ $  whether they satisfy Inequalities \eqref{eq:cond1} and \eqref{eq:cond2}.
According to Line 7, it is clear that any subset of $ \bm{q}^+ $ and $ \bm{d}^+ $ automatically satisfy Inequality \eqref{eq:cond2}.
In Lines 8$ \sim $10, we remove the pairs $ q_j \in \bm{q}^+$  and $ d_j \in \bm{d}^+$ that do not satisfy Inequality \eqref{eq:cond1}.
Note that the values of $ q $ and $ d $ (recall that {\footnotesize $ q=\sum{\bm{q}^+} $} and {\footnotesize $ d=\sum{\bm{d}^+} $}) will be recalculated due to such ``update'' (deletion in Line 10).
If $ q $ and $ d $ are updated, we need to recheck each pair of $ q_j $ and $ d_j $ in the current set of $ \bm{q}^+ $ and $ \bm{d}^+ $ until every pair of them satisfies Inequality \eqref{eq:cond1}.

\vspace{-5pt}
A subtle question may arise regarding such ``update''.
In Lines 8$ \sim $10, if \textit{several} pairs of $ q_j$ and $ d_j $ do not satisfy Inequality \eqref{eq:cond1}, say, {\footnotesize $ \{q_{1}, d_{1}\} $} and {\footnotesize $ \{q_{2}, d_{2}\} $},
one may wonder if it is possible that, after removing {\footnotesize $ \{q_{1}, d_{1}\} $} from {\footnotesize $ \bm{q}^+ $} and $ \bm{d}^+ $, Inequality \eqref{eq:cond1} can be satisfied for {\footnotesize $ \{q_{2}, d_{2}\} $} due to the update of $ q $ and $ d $, i.e., {\footnotesize $ \frac{q_2}{d_2} > \frac{(q-q_1)*(e^{{\alpha}} -1) +1}{(d-d1)*(e^{{\alpha}} -1) +1} $}.
We show that this is not possible.
If {\footnotesize $ \frac{q_1}{d_1} \leq \frac{q*(e^{{\alpha}} -1) +1}{d*(e^{{\alpha}} -1) +1}$},
we have  {\footnotesize $ \frac{q*(e^{{\alpha}} -1) +1}{d*(e^{{\alpha}} -1) +1} \leq \frac{(q-q_1)*(e^{{\alpha}} -1) +1}{(d-d1)*(e^{{\alpha}} -1) +1} $}.
Hence, {\footnotesize $ \frac{q_2}{d_2} \leq \frac{q*(e^{{\alpha}} -1) +1}{d*(e^{{\alpha}} -1) +1} \leq \frac{(q-q_1)*(e^{{\alpha}} -1) +1}{(d-d1)*(e^{{\alpha}} -1) +1} $}.
Therefore, we can remove  multiple pairs of {\footnotesize $ q_j$} and {\footnotesize $ d_j $} that do not satisfy Inequality \eqref{eq:cond1} at one time (Lines 8$ \sim $10).

It is easy to know that, if $ q_i=d_i $ for each $ i\in[1,n]$, the update will be terminated with empty $ \bm{q}^+ $ and $ \bm{d}^+ $.
In this case, we have $ q=d $; hence {\small $ \mathcal{L}^B (\cdot)$} and {\small $ \mathcal{L}^F (\cdot)$} are $ 0 $.

\textbf{Algorithm complexity}.
The time complexity for solving one linear-fractional programming problem (Lines 3$ \sim $11) w.r.t. two specific rows of the transition matrix is {\small $ O(n^2) $} because Line \ref{algo:iterate} may iterate {\small $ n*(n-1) $} times in the worst case.
The overall time complexity of Algorithm \ref{algo:cal_bpl}  is {\small $ O(n^4) $}.

\vspace{-5pt}
\section{Bounding Temporal Privacy Leakage}
\label{sec:release_mechanism}
In this section, we design private data release algorithms that can be used to convert a traditional DP mechanism into one satisfying $ \alpha $-DP$_\mathcal{T} $ by allocating calibrated privacy budgets.

We first investigate the upper bound of BPL and FPL.
We have demonstrated that  BPL and FPL may accumulate and increase over time in Figure \ref{fig:tpl}.
A natural question is that: is there a limit of  BPL and FPL over time.
For {\footnotesize $ \mathcal{M}^t $} that satisfies $ \epsilon $-DP at each {\small $ t \in [1,T] $}, we know that a loose upper bound of BPL or FPL over time $ T $ is $ T\epsilon $  according to Remark \ref{remark}.
When $ T $ is unknown, giving the upper bound of BPL or FPL is a challenge.

\begin{theorem}
\label{thm:sup}
Given a transition matrix {\footnotesize $ P_i^B $}  (or {\footnotesize $ P_i^F $}) representing temporal correlation,
let $ q $ and $ d $ be the ones that give the maximum value in Equation \eqref{eq:lfp_sol_b} (or Equation \eqref{eq:lfp_sol_f}) and {\small $ q \neq d $}.
For {\small $ \mathcal{M}^t $} that satisfies $ \epsilon $-DP at each {\small $ t \in [1,T] $},
there are four cases regarding the supremum of BPL (or FPL) over time.

\begin{myAlignSS}
\begin{cases}
 {\tiny \log \frac{{\sqrt { 4d{e^\epsilon }(1-q ) + {{( d + q{e^\epsilon } -1)}^2} }  + d + q{e^\epsilon } - 1}}{{2d}} }&   d \neq 0  \\
\log \frac{(1-q) e^\epsilon}{1-q e^\epsilon }     &  {\footnotesize d=0 \text{ and }  q \neq 1 \text{ and }  \epsilon \leq \log(1/q)} \\
\text{not exist }  &  {\footnotesize d=0 \text{ and }   q \neq 1 \text{ and }  \epsilon > \log(1/q)}  \\
\text{not exist } &  d=0  \text{ and }  q=1 
\end{cases} \nonumber
\end{myAlignSS} 
\end{theorem}
\vspace{-10pt}

We omit the proof due to space limitations.
\vspace{5pt}

The above theorem  is applicable for both BPL and FPL because the calculation of BPL and FPL is the same.
According to the previous analysis, we can consider that the growth of BPL and FPL is in the same manner but in the reversed directions on the timeline (see Figure \ref{fig:tpl}(a)(b)).

\begin{figure}[t]
\centering
\includegraphics[scale=0.55]{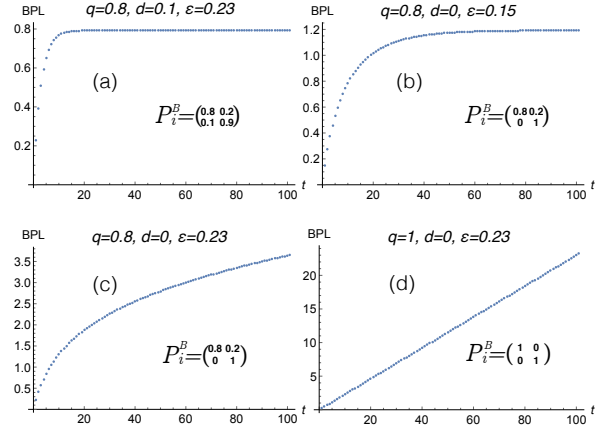} 
\caption{Examples of the maximum BPL over time.}
\label{fig:example_sup}
\vspace{-8pt}
\end{figure}

\begin{example}[The supremum of the increased BPL over time]
Suppose that $ \mathcal{M}^t $ satisfies $ \epsilon $-DP at each time point.
In Figure \ref{fig:example_sup}, we demonstrate the maximum BPL w.r.t. different  $ \epsilon $ and different transition matrices that represent {\small $ P_i^B $}.
In (a) and (b), the supremum does not exist.
In (c) and (d), we can directly calculate the supremum using Theorem \ref{thm:sup}.
The results are in line with the ones from computing BPL step by step at each time point using Algorithm \ref{algo:cal_bpl}.
\end{example}

\vspace{-5pt}
\textbf{Achieving $ \alpha $-DP$ _\mathcal{T} $ by limiting upper bound.}
We now design a data release algorithm utilizing Theorem \ref{thm:sup} to bound TPL.
Theorem \ref{thm:sup} tells us that, 
if it is not the strongest temporal correlation (i.e., {\small $ d=0 $} and {\small $ q=1 $}), 
we may bound  BPL or FPL within a desired value by allocating an appropriate privacy budget to a traditional DP mechanism at each time point.
A problem is that, in Theorem \ref{thm:sup}, the $ q $ and $ d $  are assumed to be the ones that give the maximum value of the objective function; however, they are initially unknown.
According to our analysis of Algorithm \ref{algo:cal_bpl}, such $ q $ and $ d $ depend on not only the given transition matrix but also the previous BPL (or the next FPL); however, the ``previous BPL'' is not clear when BPL achieves its supremum at some time point.
To solve this problem, we can consider that, if {\small $ T $} is approaching infinite,  BPL at time {\small $ T $} and {\small $ T+1 $} are both supremum, so that we can find $ q $ and $ d $ that give the maximum objective function using Algorithm \ref{algo:cal_bpl} by setting the ``previous BPL'' to such supremum.
Now, we can find an appropriate $ \epsilon $ to bound BPL based on Theorem \ref{thm:sup}.
For example, if {\small $ d=0 $} and {\small $ q \neq 1 $},  we can solve an equation with one variable {\footnotesize $ \epsilon $}:
{\footnotesize $ \log \frac{(1-q) e^\epsilon}{1-q e^\epsilon }   =  \alpha^B  $} (we can prove that a positive solution always exists) 
where $ \alpha^B $ is a desirable privacy level.
Similarly, we can restrict FPL within a given value.
We use this idea to bound both BPL and FPL, as shown in Algorithm \ref{algo:release1}.

\begin{algorithm}[h]
\begin{scriptsize}
\SetKwRepeat{doWhile}{do}{while}
\SetKwComment{tcc}{//}{}
\DontPrintSemicolon 
\KwIn{ { $  P_i^B$ and $  P_i^F$ , $ i \in \bm{U} $};  $ \alpha $ (desired privacy level). }
\KwOut{private data satisfying $ \alpha $-DP$_\mathcal{T} $}
  \ForEach{ user $ i  \in \bm{ U}$  }
     {  
	    Initialize $ \alpha ^B \in (0, \alpha ] $ as the supremum of BPL; \;
	    Find $ q^B $ and $ d^B $ that give the maximum value in Eq.\eqref{eq:lfp_sol_b}  using Algorithm \ref{algo:cal_bpl} with input {\tiny $ \alpha_{t-1}^B = \alpha ^B  $};  \;
	    Find $ \epsilon_i^B $ using Theorem \ref{thm:sup} with the above $ q^B $, $ d^B $ and $ \alpha ^B  $;\;
	    $ \alpha ^F \leftarrow \alpha - \alpha ^B+\epsilon_i^B $; \tcc{see Equation \eqref{eq:tpl_comp1}} 
	    Find $ q^F $ and $ d^F $ using Algorithm \ref{algo:cal_bpl} with input {\tiny $ \alpha_{t+1}^F = \alpha ^F  $}  \;
	    Find $ \epsilon_i^F $ using Theorem \ref{thm:sup} with the above $ q^F $, $ d^F $ and $ \alpha ^F  $; \;
	    \lIf{{\tiny $ \epsilon_{i}^B <  \epsilon_{i}^F $}}{ \textbf{goto} Line 2, initialize a larger $ \alpha ^B $;}
	    	   \lElseIf{{\tiny $ \epsilon_{i}^B > \epsilon_{i}^F $}}{ \textbf{goto} Line 2, initialize a smaller $ \alpha ^B $; } 
	    	   \lElse{  {\footnotesize $ \epsilon_{i} \leftarrow  \epsilon_{i}^B$ }}
     }
     $ \epsilon \leftarrow \min \{\epsilon_i, i \in \bm{U}\} $; \;
\Return{ $ \epsilon$-DP data at each time point}  
\caption{{\small Releasing Data with $ \alpha $-DP$ _\mathcal{T} $ by upper bound}}
\label{algo:release1}
\end{scriptsize}
\end{algorithm}

\textbf{Achieving $ \alpha $-DP$ _\mathcal{T} $ by privacy leakage quantification.}
We now design Algorithm \ref{algo:release2} to overcome a drawback of Algorithm \ref{algo:release1}:
when $ T $ is too short for the accumulation of temporal privacy leakage to result in a significant increase, we may not take full advantage of the privacy budgets.
Our observation is that, the DP mechanisms at the first and last time points should be allocated more budgets because they are relatively more ``influential'' in term of privacy loss.
For example,
BPL of {\footnotesize $ \mathcal{M}^t, t \in[2,T]$} is affected by the first mechanism {\footnotesize $ \mathcal{M}^1$}, and FPL of  {\footnotesize $ \mathcal{M}^t, t \in[1,T-1]$} is affected by the last mechanism {\footnotesize $\mathcal{M}^T $}.
Our idea is to allocate more privacy budgets to {\footnotesize $ \mathcal{M}^1 $} and {\footnotesize $ \mathcal{M}^T $} so that both BPL and FPL are bounded in given values at each time point.
For example, if we want that BPL at every time points are exactly the same value {\footnotesize $ \alpha^B $}, i.e., {\footnotesize $  BPL(\mathcal{M}^1)=\cdots= BPL(\mathcal{M}^T) = \alpha^B$}, then we need to make sure: (i) {\footnotesize $ PL_0(\mathcal{M}^1)=\alpha^B$}  and (ii) {\footnotesize $ \mathcal{L}^B\big( \alpha^B\big) + \epsilon_m^B = \alpha^B$} in which {\footnotesize $ \epsilon_m^B $} is the privacy budget allocated in the ``middle'' of the timeline, i.e., from $ 2 $ to $ T-1$.
We can solve the above equations to obtain   $ \epsilon_m^B $ ensuring {\footnotesize $  BPL(\mathcal{M}^1)=\cdots= BPL(\mathcal{M}^T) = \alpha^B$}.
Similarly, we can bound FPL in a given $ \alpha^F $ by finding another {\footnotesize $ \epsilon_m^F $}.
If {\footnotesize $ \epsilon_m^B \neq \epsilon_m^F $}, we can assign {\footnotesize $ \epsilon_m $} as {\footnotesize $\min \{ \epsilon_m^B, \epsilon_m^F  \} $} to ensure both BPL and FPL are bounded in $\min \{\alpha^B, \alpha^F\}  $.
It is easy to know that, when {\footnotesize $ \epsilon_m^B = \epsilon_m^F $}, we can exactly achieve $ \alpha $-DP$ _\mathcal{T} $ at each time point.
\begin{algorithm}[h]
\begin{scriptsize}
\SetKwRepeat{doWhile}{do}{while}
\SetKwComment{tcc}{//}{}
\DontPrintSemicolon 
\KwIn{ { $  P_i^B$ and $  P_i^F$ , $ i \in \bm{U} $};  $ \alpha $ (desired privacy level). }
\KwOut{private data satisfying $ \alpha $-DP$_\mathcal{T} $}
  \ForEach{ user $ i  \in \bm{ U}$  }
     {  
	    Initialize $ \alpha ^B \in (0, \alpha ] $ as the supremum of BPL, $ \epsilon_{i,1} =\alpha ^B $; \;
	    Find $ q^B $ and $ d^B $  using Algorithm \ref{algo:cal_bpl} with input {\tiny $ \alpha_{t-1}^B = \epsilon_{i,1}  $} \;
	    Find {\footnotesize $ \epsilon_{i,m}^B $} by solving {\tiny $ \mathcal{L}^B \big(\epsilon_{i,1}\big)+\epsilon_{i,m}^B =\epsilon_{i,1} $} with  {\footnotesize $ q^B, d^B$};\;
	    {\footnotesize $ \epsilon_{i,T} \leftarrow \alpha - \epsilon_{i,1}^B+\epsilon_{i,m}^B $}; \tcc{see Equation \eqref{eq:tpl_comp1}} 
	    Find {\footnotesize $ q^F $} and {\footnotesize $ d^F $} using Algorithm \ref{algo:cal_bpl} with input {\tiny $ \alpha_{t+1}^F = \epsilon_{i,T}  $} \;
	   Find {\tiny $ \epsilon_{i,m}^F $} by solving {\tiny $ \mathcal{L}^F \big(\epsilon_{i,T}\big)+\epsilon_{i,m}^F =\epsilon_{i,T} $} with  {\footnotesize $ q^F, d^F$};\;
	   \lIf{{\tiny $ \epsilon_{i,m}^B <  \epsilon_{i,m}^F $}}{ \textbf{goto} Line 2, initialize a larger $ \alpha ^B $;}
	   \lElseIf{{\tiny $ \epsilon_{i,m}^B > \epsilon_{i,m}^F $}}{ \textbf{goto} Line 2, initialize a smaller $ \alpha ^B $; } 
	   \lElse{  {\footnotesize $ \epsilon_{i,m} \leftarrow  \epsilon_{i,m}^B$}}
     }
     {\tiny $ \epsilon_1 \leftarrow \min \{\epsilon_{i,1}, i \in \bm{U}\}, \epsilon_t \leftarrow \min \{\epsilon_{i,m}, i \in \bm{U}\}, \epsilon_T \leftarrow \min \{\epsilon_{i,T}, i \in \bm{U}\}   $}; \;
\Return{ $ \epsilon_t$-DP data at $ t \in [1,T] $}  
\caption{{\small Releasing Data with $ \alpha $-DP$ _\mathcal{T}$ by quantification}}
\label{algo:release2}
\end{scriptsize}
\end{algorithm}

\vspace{-8pt}
We note that, the initializations of  {\footnotesize $ \alpha^B $} in both Algorithm \ref{algo:release1} and \ref{algo:release2} are nontrivial: too large or too small {\footnotesize $ \alpha^B $} results in more iterations to converge to {\footnotesize $ \epsilon_{i,m}^B =  \epsilon_{i,m}^F$}.
We can prove that {\footnotesize $ \epsilon_{i,m}^B =  \epsilon_{i,m}^F$} always can be achieved. 
We delegate the detailed descriptions and proofs to the long version of our paper.

\vspace{-5pt}
\section{Experimental Evaluation}
\label{sec:exp}
In this section, we design experiments for the following:
(1) verifying the runtime and correctness of our privacy leakage quantification algorithm (Algorithm \ref{algo:cal_bpl}),
(2) investigating the impact of the temporal correlations on privacy leakage and
(3) evaluating the data release Algorithms \ref{algo:release1} and \ref{algo:release2}.
We implemented all the algorithms in Java and conducted the experiments on a machine with an Intel Core i7 2.8GHz  CPU and 16 GB RAM running OSX El Capitan.

\vspace{-2pt}
\textbf{The setting of temporal correlations.}
To evaluate if our privacy loss quantification algorithms can perform well under diverse circumstances, we need different degrees of temporal correlations.
Although there are well studied methods to estimate the temporal correlations, in our experiments, we \textit{generate} the correlations (transition matrices) directly to eliminate the effect of different estimation algorithms or datasets.

\vspace{-2pt}
We now present a method for obtaining different degrees of temporal correlations.
First, we generate a transition matrix indicating the ``strongest'' correlation that contains a cell with probability $ 1.0 $ at each row but for different columns (this type of transition matrix will lead to an upper bound of TPL, as shown in Examples \ref{example:bpl} and \ref{example:fpl}).
Then, we perform \textit{Laplacian smoothing}\cite{sorkine_laplacian_2004}, which is a method originally used to smooth a polygonal mesh, to \textit{uniformize} the probabilities of $ P_i $ in different degree. 
Let $ {p_{jk}}  $ be an element at the $ j $th row and $ k $th column of the matrix $ P_i $.
The new probabilities $ \hat{p_{jk}}  $ are generated using Equation \eqref{eq:laplace_smoothing}, where $ s $ is a positive parameter that controls the degrees of smoothing.
A smaller $s$ results in a stronger temporal correlation.
\begin{myAlignSSS}
\hat{p_{jk}} = \frac{ p_{jk} +s }{\sum_{u=1}^{n}{(p_{ju} +s)}} \label{eq:laplace_smoothing}
\end{myAlignSSS}
\vspace{-11pt}
We note that, the degrees of correlation with $ s $ are only comparable with each other under the same $ n $ (i.e., $ |\bm{loc}| $). 

%

\subsection{Runtime of Privacy Quantification Algorithms}
\label{subsec:runtime}
In this section, we compare the runtime of our algorithm  with  \textsf{Gurobi}\footnote{http://www.gurobi.com/.  Commercial software. We use version 6.5.}  and \textsf{lp\_solve}\footnote{http://lpsolve.sourceforge.net/. Open source software. We use version 5.5.}, which are two well-known softwares for solving optimization problems, e.g., the linear-fractional programming problem \eqref{eq:lfp}$ \sim $\eqref{eq:lfp_end2} in our setting.
We run our privacy quantification algorithm $ 30 $ times, and run \textsf{Gurobi} and \textsf{lp\_solve} $ 5 $ times (because they are very time-consuming),
and then calculate the average runtime for each of them.
At each time, we randomly generate a transition matrix $ P_i $ whose elements are uniformly drawn from $ [0,1] $. 
We verified that the optimal solution returned by the three algorithms are the same.
In the following, we describe two factors that may affect the runtime: $ \alpha $ as BPL at the previous time point or FPL at the next time point (i.e., one input of Algorithm \ref{algo:cal_bpl}), $ n $ as the domain size of transition matrix.
The results are shown in Figure \ref{fig:ex_runtime}.

\begin{figure}[t]
\centering
\includegraphics[scale=0.35]{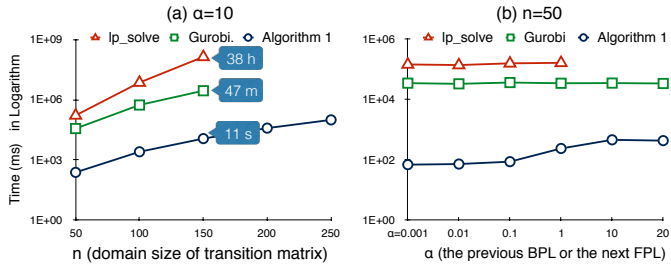} 
\caption{Runtime of Privacy Quantification Algorithms.}
\label{fig:ex_runtime}
\vspace{-10pt}
\end{figure}

\textbf{Runtime vs. $ n $}.
In Figure \ref{fig:ex_runtime}(a), we show the runtime of the three algorithms with inputs of {\normalsize $ \alpha=10 $} and a {\normalsize $ n \times n $} random probability matrix $ P_i $.
The runtime of all algorithms increase along with $ n $ because $ n $ is the number of variables in our linear-fractional program.
Algorithm \ref{algo:cal_bpl} significantly outperforms \textsf{Gurobi} and \textsf{lp\_solve}.
For example, in Figure \ref{fig:ex_runtime}(a), when {\normalsize $ n=150 $}, Algorithm \ref{algo:cal_bpl} only spends $ 11 $ seconds, whereas the runtime of \textsf{Gurobi} and \textsf{lp\_solve} are about $ 47 $ minutes and {$ 38 $} hours, respectively.
Since \textsf{Gurobi} and \textsf{lp\_solve} spend tremendous time when {\normalsize $ n>150 $},  we omit them in the graph. 


\textbf{Runtime vs. $ \alpha $}.
In Figure \ref{fig:ex_runtime}(b), we show that, a larger previous BPL (or the next FPL), i.e., $ \alpha $, may lead to higher runtime of Algorithm \ref{algo:cal_bpl}, whereas \textsf{Gurobi} and \textsf{lp\_solve} are stable for varying $ \alpha $.
The reason is that, when $ \alpha $ is large, 
Algorithm \ref{algo:cal_bpl} may take more time in 
Lines $ 9 $ and $ 10 $ for updating  each pair of $ q_j  \in \bm{q}^+$ and $ d_j  \in \bm{d}^+$ to satisfy Inequality \eqref{eq:cond1}.
An update in Line $ 10 $ is more likely to occur due to a large $ \alpha $ because {\footnotesize $  \frac{q(e^{{\alpha}} -1) + 1 }{d(e^{{\alpha}} - 1) + 1} $} is increasing with $ \alpha $.
However, such growth of runtime along with $ \alpha $ will not last so long because the update happens $ n-1 $ times in the worse case (according to our previous analysis, the update will be terminated if only one element is left in $ \bm{q}^+ $).
As shown in Figure \ref{fig:ex_runtime}(b), when $ \alpha >10 $, the runtime of Algorithm \ref{algo:cal_bpl} becomes stable.
We only obtain a part of the runtime for \textsf{lp\_solve}  because a precision problem occurs when $ \alpha \geq 10 $ due to the design of \textsf{lp\_solve}.

\subsection{Impact of Temporal Correlations on Privacy Leakage}
\label{subsec:ex_impact}

In this section, for the convenience of explanation, we only present the impact of temporal correlations on BPL because
 the growth of BPL and FPL are in the same way but in the reversed directions on the timeline.
We examined  $ s $ values in Equation \eqref{eq:laplace_smoothing} ranging from 0.005 to 1. 
We set $n $ to $ 50 $ and $ 200 $. 
Let $ \varepsilon $ be the privacy budget of $ \mathcal{M}^t $ at each time point. 
We test $ \varepsilon =1$ and 0.1. 
The results are shown in Figure \ref{fig:ex_tpl} and  are summarized as follows.

\textbf{Privacy Leakage vs. $ \bm{s} $.}
Figure \ref{fig:ex_tpl} shows that the privacy leakage caused by a non-trivial temporal correlation will increase over time, and such growth first increases sharply and then remains stable because the increment is calculated recursively.
The increase caused by a stronger temporal correlations (i.e., smaller $ s $) is steeper, and the time for the increase is longer.
Consequently, stronger correlations result in higher privacy leakage.

\textbf{Privacy Leakage vs. $ \bm{\varepsilon} $.}
Comparing Figures \ref{fig:ex_tpl}(a) and (b), we found that  $0.1$-DP  significantly delayed the growth of privacy leakage.
Taking $ s=0.005 $, for example, the noticeable increase continues for almost 8 timestamps when $ \varepsilon=1 $ (Figures \ref{fig:ex_tpl}(a)), whereas it continues for approximately 80 timestamps when $ \varepsilon=0.1 $ (Figures \ref{fig:ex_tpl}(b)).
However, after a sufficient long time, the privacy leakage in the case of $ \varepsilon=0.1 $ is not substantially lower than that of $ \varepsilon=1 $ under stronger temporal correlations.  
This is because, although the privacy leakage is eliminated at each time point by setting a small privacy budget, the adversaries can eventually learn sufficient information from the continuous releases.

\textbf{Privacy Leakage vs. $ n$.}
Under the same $ s $, TPL is smaller when $ n$ (dimension of the transition matrix) is larger, as shown in the lines $ s=0.005 $ with $ n =50$ and $ n =200$ of Figure \ref{fig:ex_tpl}. 
 This is because the transition matrices tend to be uniform (weaker correlations) when the dimension is larger.

In conclusion, the experiments reveal that our quantification algorithms can flexibly respond to different degrees of temporal correlations.

\begin{figure}[t]
\centering
\includegraphics[width=0.95\linewidth]{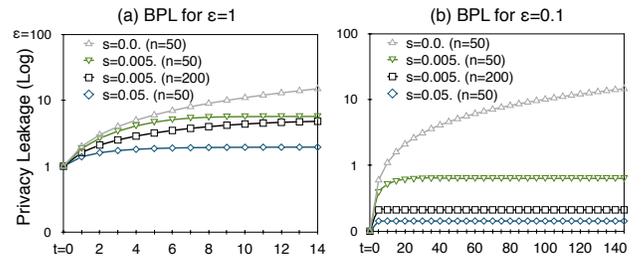} 
\caption{Evaluation of BPL.}
\label{fig:ex_tpl}
\vspace{-10pt}
\end{figure}

\begin{figure}[t]
\centering
\includegraphics[scale=0.33]{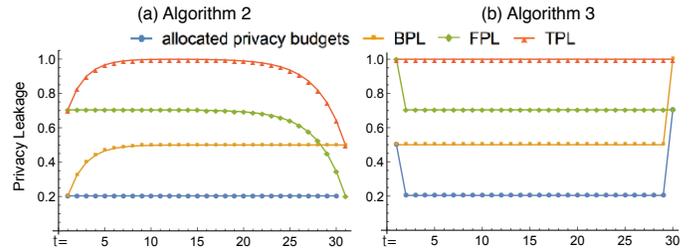} 
\caption{Data Release Algorithms with {\small $ 1 $-DP$_\mathcal{T} $}.}
\label{fig:ex_budgets}
\vspace{-10pt}
\end{figure}

\subsection{Evaluation of Data Releasing Algorithms}
In this section, we first show a visualization of privacy allocation of  Algorithms \ref{algo:release1} and \ref{algo:release2}, then we compare the data utility in terms of Laplace noise.

Figure \ref{fig:ex_budgets} shows an example of budget allocation, w.r.t. {\scriptsize $ P_i^B=\big(\begin{smallmatrix} 0.8&0.2\\ 0.2&0.8\end{smallmatrix} \big)$} and {\scriptsize $ P_i^F=\big(\begin{smallmatrix} 0.8&0.2\\ 0.1&0.9\end{smallmatrix} \big)$}.
The goal is $ 1 $-DP$ _\mathcal{T} $.
It is easy to see that Algorithm \ref{algo:release2} has better data utility because it exactly achieves the desired privacy level.

Figure \ref{fig:ex_utility} shows the data utility of Algorithms \ref{algo:release1} and \ref{algo:release2} with $ 2 $-DP$ _\mathcal{T} $.
We calculate the absolute value of the Laplace noise with the allocated budgets (as shown in Figure \ref{fig:ex_budgets}).
Higher value of noise indicates lower data utility.
In Figure \ref{fig:ex_utility}(a), we test the data utility under backward and forward temporal correlation both with parameter $ s=0.001 $, which means relatively strong correlation.
It shows that, when $ T $ is short, Algorithm \ref{algo:release2} outperforms Algorithm \ref{algo:release1}.
In other words, regardless of how long $ T $ is,  Algorithm \ref{algo:release1} perturbs data in the same way.
In Figure \ref{fig:ex_utility}(b), we investigate the data utility under different degree of correlations.
The dash line indicates the absolute value of Laplace noise if no temporal correlation exists (privacy budget is $ 2 $).
It is easy to see that the  data utility significantly decays under strong correlation $ s=0.01 $.

\begin{figure}[t]
\centering
\includegraphics[scale=0.27]{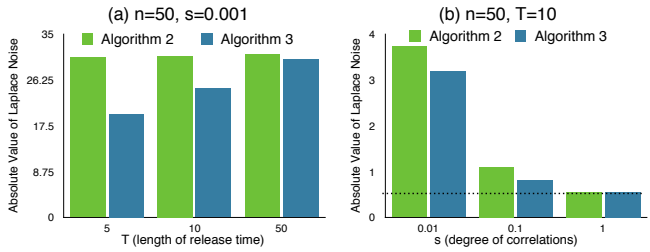} 
\caption{Data utility of $ 2 $-DP$ _\mathcal{T} $ mechanisms.}
\label{fig:ex_utility}
\vspace{-8pt}
\end{figure}

\vspace{-10pt}
\section{Related Work}
\vspace{-3pt}
Several studies have questioned whether differential privacy is valid for correlated data. 
Kifer and Machanavajjhala\cite{kifer_no_2011}\cite{kifer_rigorous_2012}\cite{kifer_pufferfish:_2014} first raised the important issue that differential privacy may not guarantee privacy if  adversaries know the data correlations.
In their line of work, they\cite{kifer_no_2011} argued that it is not possible to ensure any utility in addition to privacy without making assumptions about the data-generating distribution and the background knowledge available to an adversary.
To this end, they proposed a general and customizable privacy framework called \textit{PufferFish}, in which the potential secrets, discriminative pairs, and data generation need to be explicitly defined. 
Yang et al.\cite{yang_bayesian_2015} further investigated differential privacy on correlated tuples described using a proposed Gaussian correlation model.
The privacy leakage w.r.t. adversaries with specified prior knowledge can be efficiently computed.

Zhu et al. \cite{zhu_correlated_2015} proposed correlated differential privacy by redefining the sensitivity of queries on correlated data; however, the privacy guarantee provided by this definition for spatio-temporal data is unclear.
Very recently, Liu et al. \cite{changchang_liu_dependence_2016} proposed dependent differential privacy by introducing dependence coefficients for analyzing the sensitivity of different queries under probabilistic dependences between tuples.
However, such dependence coefficients do not easily account for the spatio-temporal correlations.

Dwork et al. first studied \textit{differential privacy under continual observation} and proposed event-level/user-level privacy\cite{dwork_differential_2010-2}\cite{dwork_differential_2010}.
The previous studies in this setting focused on the problems of 
high dimension\cite{acs_case_2014}
\cite{li_differentially_2014}\cite{xiao_dpcube:_2014}, 
infinite sequence\cite{cao_differentially_2015}\cite{cao_differentially_2016}\cite{kellaris_differentially_2014},
sliding window queries\cite{cao_efficient_2013},
and real-time publishing\cite{fan_fast:_2013}.
\cite{haoran:_2015}. 
None of them addressed the problem of  temporally correlated data.

To the best of our knowledge, no study has reported the risk of differential privacy under temporal correlations for the continuous aggregate release setting.
Although a few studies\cite{shokri_quantifying_2011}\cite{xiao_protecting_2015} have considered a similar adversarial model in which the adversaries have prior knowledge of temporal correlations represented by Markov chains, they focused on location privacy in the single-user setting.
Shokri et al.\cite{shokri_quantifying_2011} proposed an evaluation framework for location privacy protection, assuming that the adversary knows the transition probabilities of each user.
Xiao et al.\cite{xiao_protecting_2015} proposed a mechanism extending DP for single user location sharing  under temporal correlations modeled by Markov chains.
In contrast, the scenario in this paper focuses on quantifying the privacy loss of traditional DP mechanisms under temporal  correlations for continuous aggregate release setting.

\vspace{-5pt}

\section{Conclusions}
In this paper, we quantified the risk of differential privacy under temporal  correlations by formalizing, analyzing and calculating the privacy loss against adversaries who have varying degrees of temporal  correlations.
This work opens up interesting future research directions, such as  modeling temporal correlations with other type of correlations (e.g. tuple-wise correlations), and combining our methods with the previous studies that neglected the effect of temporal  correlations in order to bound the temporal privacy leakage.

\section{Acknowledgment}
This work was supported by JSPS KAKENHI Grant Number 16K12437, 
the National Institute of Health (NIH) under award number R01GM114612, 
the Patient-Centered Outcomes Research Institute (PCORI) under contract ME-1310-07058, 
and the National Science Foundation under award CNS-1618932.



\bibliographystyle{abbrv}
\bibliography{./002_Correlation_icde17}

%


\appendices

\vspace{-10pt}

\section{Proof of Theorem \ref{thm:lfp}}
\label{appex:thm_tpl}
We need Dinkelbach's Theorem and Lemma \ref{lem:lfp_proof}  in our proof.
\begin{theorem}[Dinkelbach's Theorem\cite{dinkelbach_nonlinear_1967}]
\label{thm:Dinkelbach}
In a linear-fractional programming problem, suppose that the variable vector is $\boldsymbol{x}$ and the objective function is represented as {\small $\frac{Q(\boldsymbol{x})}{D(\boldsymbol{x})}$}.
Vector $\boldsymbol{x^*}$ is an optimal solution if and only if 
 \begin{myAlignSSS}
\max\{Q(\boldsymbol{x})-\lambda*D(\boldsymbol{x})\} = 0 
\text{ where } \lambda= \frac{Q(\boldsymbol{x^*})}{D(\boldsymbol{x^*})}.  \label{eq:dink}
\end{myAlignSSS}
\end{theorem}
\vspace{-8pt}
\begin{lemma}
\label{lem:lfp_proof}
For the following maximization problem ({\small $ k_1,...,k_n\in \mathbb{R}$}) with the same constraints as the ones in the linear-fractional programming \eqref{eq:lfp}$ \sim $\eqref{eq:lfp_end2},
\begin{myAlignSS}
 \textnormal{maximize } &  \text{ } k_1 x_1+\cdots+k_{n} x_{n}  \nonumber  \\
 \textnormal{subject to  } 
  & \text{ } e^{-{{\alpha}^B_{t-1}}}*{x_k}  \leq {x_j}\leq e^{{{\alpha}^B_{t-1}}}*{x_k} ,   \nonumber \\
  &\text{ }  0 < {x_j}< 1 \text{ and }  0<x_k<1, \nonumber  \\
  & \text{ where } x_j,x_k \in \bm{x}, \text{ } j,k \in  [1,n]. \nonumber
\end{myAlignSS} 
\vspace{-15pt}
\noindent
an optimal solution is as follows: if  {\small $ k_i>0 $}, let {\small $ x_i=e^{{{\alpha}^B_{t-1}}} m$} where $ m$ is a positive real number; if  {\small $ k_i\leq 0 $}, let {\small $  x_i=m$}.
\end{lemma}
\vspace{-5pt}
\begin{proof}
Without loss of generality, we suppose that the smallest value in the optimal solution is $ x_{n}  $.
Let $ y_j $ be {\footnotesize $\frac{ x_j}{x_{n} }$} for {\footnotesize $ j\in[1,n-1] $}; then, {\footnotesize $ 1 \leq y_j \leq {e^{{\alpha}^B_{t-1}}} $}.
Replacing $ x_j $ with $ y_j $ and setting {\footnotesize $ x_{n} =m$}, we have a new objective function {\footnotesize $ \frac{1}{m} * ( k_1 y_1+\cdots+k_{n-1} y_{n-1}+k_{n} ) $} whose solution is equivalent to the original one.
Because the only constraint is {\footnotesize $ 1 \leq y_j \leq e^{{{\alpha}^B_{t-1}}} $}, the following is an optimal solution for the maximum objective function: if {\footnotesize $ k_j>0 $}, let {\footnotesize $ y_j =e^{{{\alpha}^B_{t-1}}} $}; if {\footnotesize $ k_j \leq 0 $}, let {\small $ y_j=1$}. 
\end{proof}

\begin{proof}[Proof of Theorem \ref{thm:lfp}]
We first prove that, under the conditions shown in Theorem \ref{thm:lfp}, i.e., Inequalities \eqref{eq:cond1} and \eqref{eq:cond2}, an optimal solution of the problem \eqref{eq:lfp}$ \sim $\eqref{eq:lfp_end2} is:
\vspace{-2pt}
\begin{myAlignSSS}
\bm{x}^* =
\begin{cases}
 x_j = e^{{{\alpha}^B_{t-1}}}*m & x_j \in \bm{x}^+ \\
x_k = m & x_k \in \bm{x}^- \label{eq:opt_sol}
\end{cases},
\end{myAlignSSS}
\vspace{-5pt}
where $ m $ is a positive real number.

\vspace{-2pt}
For convenience, we rewrite our objective function as {\footnotesize $\frac{Q(\bm{x})}{D(\bm{x})}  $} in which {\footnotesize $ Q(\bm{x})=\bm{q}\bm{x} $} and  {\footnotesize $ D(\bm{x})=\bm{d}\bm{x} $}.
Substituting {\footnotesize $ \bm{x}^* $} of Equation \eqref{eq:opt_sol} into {\footnotesize $ Q(\bm{x}) $} and {\footnotesize $ D(\bm{x}) $}, we have {\footnotesize $ Q(\bm{x}^*)=q(e^{{{\alpha}^B_{t-1}}}-1)+1 $} and {\footnotesize $ D(\bm{x}^*)=d(e^{{{\alpha}^B_{t-1}}}-1)+1 $} (recall that {\footnotesize $ q=\sum{\bm{q}^+} $} and {\footnotesize $d=\sum{\bm{d}^+} $}).
Then, we can rewrite Inequalities \eqref{eq:cond1} and \eqref{eq:cond2} 
in Theorem \ref{thm:lfp} 
as follows.
\begin{myAlignSSS}
&\frac{q_j}{d_j} > \frac{Q(\bm{x}^*)}{D(\bm{x}^*)}, & \forall j\in[1,n] \text{ where } q_j \in \bm{q}^+, d_j \in \bm{d}^+ \label{eq:1}\\
&\frac{q_k}{d_k} \leq \frac{Q(\bm{x}^*) }{D(\bm{x}^*)}, & \forall k\in[1,n] \text{ where } q_k \in \bm{q}^-, d_k \in \bm{d}^- \label{eq:2}
\end{myAlignSSS}

\vspace{-8pt}
According to Dinkelbach's Theorem, to prove $ \bm{x}^* $ in \eqref{eq:opt_sol} is an optimal solution,  we only need to prove the following equation because of {\footnotesize $ D(\bm{x}^*)>0 $}. 
\begin{myAlignSS}
\text{maximize } \{D(\bm{x}^*)  Q(\bm{x}) - Q(\bm{x}^*)  D(\bm{x}) \}=0.  \label{eq:thm_max}
\end{myAlignSS}
\vspace{-10pt}
We expand the above equation as follows.
\begin{myAlignSSS}
&\text{Eqn.}\eqref{eq:thm_max}=D(\bm{x}^*)  (\bm{q}^+ \bm{x}^+ + \bm{q}^- \bm{x}^-) 
 - 
 Q(\bm{x}^*)  (\bm{d}^+ \bm{x}^+ + \bm{d}^- \bm{x}^-)  \nonumber \\
 &= \big(D(\bm{x}^*)  \bm{q}^+ - Q(\bm{x}^*)  \bm{d}^+\big) \bm{x}^+
 + 
 \big(D(\bm{x}^*)  \bm{q}^- - Q(\bm{x}^*)   \bm{d}^-\big) \bm{x}^-   \label{eq:last}  
\end{myAlignSSS}

\vspace{-10pt}
By Equations \eqref{eq:1} and \eqref{eq:2}), we have
{\footnotesize $ D(\bm{x}^*)  \bm{q}^+ - Q(\bm{x}^*)  \bm{d}^+ >0$} and {\footnotesize $ D(\bm{x}^*)  \bm{q}^- - Q(\bm{x}^*)   \bm{d}^- \leq 0 $}.
Hence, according to Lemma \ref{lem:lfp_proof}, we can obtain the maximum value in Equation \eqref{eq:thm_max} by setting {\footnotesize $ \bm{x}^+=[e^{{\alpha}^B_{t-1}}*m] $} and {\footnotesize $ \bm{x}^-=[m] $} where {\footnotesize $ m $} is a positive real number.
Now, we obtain the maximum value in Equation \eqref{eq:thm_max}.
\begin{myAlignSS}
&\big((D(\bm{x}^*) q - Q(\bm{x}^*) d)  e^{\varepsilon} m
 + 
\big(D(\bm{x}^*)  (1-q) - Q(\bm{x}^*)  (1-d)) \big) m \nonumber \\
=& \big(D(\bm{x}^*) (q e^\varepsilon+(1-q)) - Q(\bm{x}^*) (d  e^\varepsilon +(1-d))\big)m  \nonumber \\
=& \big(D(\bm{x}^*) Q(\bm{x}^*) - Q(\bm{x}^*) D(\bm{x}^*)\big) m = 0  \nonumber 
\end{myAlignSS}

\vspace{-10pt}
\noindent
Therefore, by Dinkelbach's Theorem, $ \bm{x}^* $ is an optimal solution for the problem \eqref{eq:lfp}$ \sim $\eqref{eq:lfp_end2}.
Substituting them into the objective function \eqref{eq:lfp}, we obtain the maximum value 
 {\footnotesize $ \frac{q(e^{{\alpha}_{t-1}^B} -1) + 1 }{d(e^{{\alpha}_{t-1}^B} - 1) + 1} $}.
\end{proof}

\end{document}